\documentclass{ims9x6}

\usepackage{bookmark,dsfont,cite}

\DeclareMathOperator{\supp}{supp} 
\DeclareMathOperator{\Tr}{Tr} 

\def\bq{\begin{eqnarray}}
\def\eq{\end{eqnarray}}

\def\nn{\nonumber}
\def\eps{\varepsilon}
\newcommand{\norm}[1]{\left\lVert #1 \right\rVert}

\newcommand\1{{\ensuremath {\mathds 1} }}

\def\cF{\mathcal{F}}
\def\cE{\mathcal{E}}
\def\cK{\mathcal{K}}
\def\R{\mathbb{R}}
\def\d{{\rm d}}
\newcommand{\abs}[1]{\lvert#1\rvert}
\newcommand{\I}{\mathrm{i}}  

\newenvironment{Proof}[1][Proof]{%
\par\addvspace{12pt plus3pt minus3pt}\global\logotrue%
\noindent{\bf{#1}:\hskip.5em}\ignorespaces}{%
    \par\iflogo\vskip-\lastskip
    \vskip-\baselineskip\prbox\par
    \addvspace{12pt plus3pt minus3pt}\fi}

\makeatletter
\renewcommand{\ps@plain}{%
  \renewcommand{\@oddhead}{\hfil\footnotesize%
    \raisebox{30pt}[0pt][0pt]{\parbox{300pt}{\centering%
      A contribution to the Proceedings of the\\{}%
      Workshop on Density Functionals for Many-Particle Systems\\{}
      2--29 September 2019, Singapore}}\hfil}%
  \renewcommand{\@evenhead}{\@oddhead}%
  \renewcommand{\@oddfoot}{\hfil\footnotesize%
        \raisebox{-8pt}[0pt][0pt]{\thepage}\hfil}%
  \renewcommand{\@evenfoot}{\@oddfoot}%
}
\makeatother
\renewcommand{\trimmarks}{\relax}

\begin{document}

\setcounter{page}{1}  

\chapter{\uppercase{Direct methods to\\ Lieb--Thirring kinetic inequalities}}

\markboth{Phan Th\`anh Nam}%
{Direct methods to Lieb--Thirring kinetic inequalities}

\author{Phan Th\`anh Nam}

\address{Department of Mathematics, LMU Munich\\ Theresienstrasse 39, %
         80333 Munich, Germany\\ nam@math.lmu.de }

\begin{abstract}
We review some recent progress on Lieb--Thirring inequalities, focusing on
direct methods to kinetic estimates for orthonormal functions and
applications for many-body quantum systems. 
\end{abstract}

\section{Introduction}
The celebrated Weyl law states that the asymptotic behavior of negative
eigenvalues  
$$E_1(-\Delta+\lambda V)\le E_2(-\Delta+\lambda V) \le\cdots <0$$
of the Schr\"odinger operator $-\Delta +\lambda V(x)$ on $L^2(\R^d)$ can be
determined by the semiclassical approximation in the strong coupling limit
$\lambda\to \infty$, namely  
\begin{eqnarray*}
  \sum_{n\ge 1} \bigl|E_n(-\Delta+\lambda V)\bigr|^\kappa
  &\approx&\int_{\R^d}\int_{\R^d}
            \bigl|\bigl( |2\pi k|^2 + \lambda V(x) \bigr)_- \bigr|^{\kappa}
            \,\d k\, \d x\\
  &=&  L^{\rm cl}_{\kappa,d}
      \int_{\R^d} \bigl|\lambda V(x)_-\bigr|^{\kappa+d/2} \,\d x  
\end{eqnarray*}
for all $\kappa\ge 0$.
Here $t_-=\min\{t,0\}$ is the negative part of $t$ and 
\bq \label{eq:Lcl}
L^{\rm cl}_{\kappa,d} = \int_{\R^d}
\bigl|\bigl( |2\pi k|^2 + 1 \bigr)_- \bigr|^{\kappa}\, \d k =
 (4\pi)^{-\frac d 2} \frac{\Gamma(\kappa+1)}{\Gamma(\kappa+1+ \frac d 2)}. 
\eq
The case ${\kappa=0}$ corresponds to the number of negative eigenvalues.
By formally taking the  potential 
$$
V(x)= \begin{cases}
               -1        &\text{if } x\in \Omega,\\
               +\infty   &\text{if } x\notin \Omega              
            \end{cases} 
$$   
with an open bounded set $\Omega\subset \R^d$, we find that the number of
eigenvalues below $\lambda$ of the Dirichlet Laplacian on $L^2(\Omega)$ is
equal to  
$$
L^{\rm cl}_{0,d} |\Omega|\lambda^{d/2}
+ o\bigl(\lambda^{d/2}\bigr)_{\lambda\to \infty}. 
$$
The latter formula was first proved by Weyl in 1911 \cite{Weyl-11,Weyl-12}.
We refer to \cite[Chapter 12]{LieLos-01} for further discussion on Weyl's law. 

Lieb--Thirring inequalities are non-asymptotic estimates for eigenvalues of
Schr\"odinger operators which agree with the semiclassical approximation,
possibly up to a universal constant factor.  

\begin{theorem}\mbox{\normalfont(Lieb--Thirring inequalities)}\label{thm0:LT}
Let $d \ge 1$ and $\kappa \ge 0$. Let $V:\R^d \to \R$ be a real-valued
potential such that $V_- \in L^{\kappa+d/2}(\R^d)$.
Assume that the Schr\"odinger operator $-\Delta+V(x)$ on $L^2(\R^d)$ has
negative eigenvalues $\bigl\{E_n(-\Delta+V)\bigr\}_{n\ge 1}$.
Then
\bq \label{eq:intro-LT-ineq}
\sum_{n\ge 1} \bigl|E_n(-\Delta+V)\bigr|^\kappa
\le L_{\kappa,d} \int_{\R^d} \bigl|V(x)_-\bigr|^{\kappa+\frac d 2} \,\d x
\eq
for a constant $L_{\kappa,d}\in (0,\infty)$ independent of $V$, provided that
\[
		\begin{cases}
		\kappa \ge 0   &\mathrm{if}\ d\ge 3,\\
		\kappa>0       &\mathrm{if}\ d= 2,\\
		\kappa \ge 1/2 &\mathrm{if}\ d=1.
		\end{cases}	
	\]
\end{theorem}
      
This result was first proved  by Lieb and Thirring in 1975 for ${\kappa=1}$
and ${d=3}$ \cite{LieThi-75}.
Then they extended the inequality to all ${\kappa>0}$ when ${d\ge 2}$ and all
${\kappa>1/2}$ when ${d=1}$ \cite{LieThi-76}.
The case ${\kappa=0}$ when ${d\ge 3}$, often referred to as the
Cwikel--Lieb--Rozenblum inequality, was proved independently in
\cite{Cwikel-77,Lieb-76,Rozenblum-76}.
The last critical case ${\kappa=1/2}$ when ${d=1}$ was solved by Weidl in 1996
\cite{Weidl-96}.
The range of $\kappa$ is optimal. 

In general, the constant $L_{\kappa,d}$ in \eqref{eq:intro-LT-ineq} is not
necessarily the same as the semiclassical constant $L_{\kappa,d}^{\rm cl}$
in \eqref{eq:Lcl}.  
Determining the sharp Lieb--Thirring constant is an important topic in
mathematical physics; see e.g. \cite{FraGonLew-20} for a recent study.
So far the sharp value of $L_{\kappa,d}$ is only known in two cases: 
\begin{itemize}
\item $L_{\kappa,d}= L^{\rm cl}_{\kappa,d}$ for $\kappa \ge 3/2$, $d\ge 1$. It
  was proved for $d=1$ by Lieb and Thirring in 1976 \cite{LieThi-76}, and
  extended to all $d\ge 1$ by Laptev and Weidl in 2000 \cite{LapWei-00}.  
\item $L_{1/2,1}=2L_{1/2,1}^{\rm cl}$ for $\kappa=1/2$, $d=1$. It was proved
  by Hundertmark, Lieb and Thomas in 1998 \cite{HunLieTho-98}.  
\end{itemize}
We refer to \cite{Lieb-00,LieLos-01,Simon-05,LieSei-10} for pedagogical
introductions to the subject and \cite{Frank-20} for a recent review  of
current research. 

In this short note we will focus on the case $\kappa=1$ (sum of eigenvalues)
which is particularly interesting due to its application to the study of the
ground state energy of Fermi gases.
By a duality argument \cite{LieThi-75,LieThi-76}, the bound on the sum of
eigenvalues can be translated to a kinetic inequality for orthonormal
functions in $L^2(\R^d)$.

\begin{theorem}\mbox{\normalfont(Lieb--Thirring kinetic inequality)}
\label{thm:LT-kinetic-intro}
Let ${d \ge 1}$. For any ${N\ge 1}$, let $\{u_n\}_{n=1}^N$ be orthonormal
functions in $L^2(\R^d)$ and define
${\rho(x)=\sum_{n=1}^N \bigl|u_n(x)\bigr|^2}$.
Then 
\bq \label{eq:intro-LT-ineq-kin}
\sum_{n=1}^N \int_{\R^d} \bigl|\nabla u_n(x)\bigr|^2\,\d x
\ge K_d \int_{\R^d}  \rho(x)^{1+\frac {2}{d}}\,\d x.
\eq
The constant $K_d>0$ is related to the sharp constant $L_{1,d}$ in
\eqref{eq:intro-LT-ineq} by 
\bq \label{eq:L-K}
K_d \left(1+\frac{2}{d}\right) = \left[ L_{1,d} \left(1+\frac{d}{2}\right)
\right]^{-2/d}. 
\eq
\end{theorem}

Here the orthogonality of $\{u_n\}_{n=1}^N$ is crucial.
Without that assumption, we only have at best
\bq \label{eq:GN-intro}
\sum_{n=1}^N \int_{\R^d} \bigl|\nabla u_n(x)\bigr|^2\,\d x
\ge \frac{C_{\rm GN}(d)}{N^{2/d}}
\int_{\R^d}  \rho(x)^{1+\frac {2}{d}}\,\d x 
\eq
where $C_{\rm GN}(d)$ is the optimal Gagliardo--Nirenberg constant
\begin{align} \label{eq:C-GN}
	C_{\rm GN}(d):= \inf_{\substack{u \in H^1(\R^d)\\ 
  \norm{u}_{L^2(\R^d)} = 1}}
  \frac{\int_{\R^d} \bigl|\nabla u(x)\bigr|^2\,\d x}
       {\int_{\R^d} \bigl|u(x)\bigr|^{2 (1+\frac{2s}{d})}\,\d x}. 
\end{align}
In contrast, the constant $K_d$ in \eqref{eq:intro-LT-ineq-kin} is
independent of $N$, making it very useful to study quantum systems with many
particles.
On the physical point of view, while the Sobolev and the Gagliardo--Nirenberg
inequalities are quantitative versions of the uncertainty principle,
Lieb--Thirring inequalities are deeper as they involve the exclusion
principle as well.  
  
Historically, the kinetic inequality \eqref{eq:intro-LT-ineq-kin} is a key
ingredient in the short proof of the stability of matter in
\cite{LieThi-75}.
In \cite{LieThi-75}, Lieb and Thirring derived \eqref{eq:intro-LT-ineq-kin}
from the eigenvalue bound \eqref{eq:intro-LT-ineq}.
The advantage of studying \eqref{eq:intro-LT-ineq} is that it can be reduced
to counting the number of eigenvalues $\ge 1$ of the Birman--Schwinger operator
$$
K_E=\sqrt{|V_-|} (-\Delta +E )^{-1}\sqrt{|V_-|}, \quad E>0.
$$
This method has been at the heart of most semiclassical eigenvalue estimates since the
1970s.   

Surprisingly, a direct proof of the kinetic bound 
\eqref{eq:intro-LT-ineq-kin} had not been available for a long time until
the work of Eden and Foias in one dimension in 1991 \cite{EdeFoi-91}, which was
extended to all dimensions by Dolbeault, Laptev, and Loss in 2008
\cite{DolLapLos-08}. 
More recently, different proofs were found by Rumin in 2010 
\cite{Rumin-10,Rumin-11}
and  Lundholm and Solovej in 2013 \cite{LunSol-13,LunSol-14}.
These new approaches have immediately led to several interesting developments
in the field.  

\enlargethispage{1.0\baselineskip}

In this note, we will review the latter two direct approaches and explain some
new results originated from them in a rather self-consistent manner.
Here is a summary of the next sections. 
\begin{itemize}
\item
In Section \ref{sec:2} we discuss Rumin's method.
In principle, this approach is based on a suitable decomposition of the
Laplacian in the momentum space.
As a warm-up, we represent a short proof of the standard Sobolev inequality by
Chemin and Xu \cite{CheXu-97} in Section \ref{sec:2.1}.
In some sense, Rumin's idea is to combine this technique with Bessel's
inequality. 
The kinetic inequality \eqref{eq:intro-LT-ineq-kin} and its extensions are
proved in Section \ref{sec:2.2}.
We explain the implication to the eigenvalue bound \eqref{eq:intro-LT-ineq} in
Section \ref{sec:2.3}.
Interestingly, Rumin's method  can be modified to give the currently best
bound for the constant in \eqref{eq:intro-LT-ineq-kin}.
This result from \cite{FraHunJexNam-18}  will be reviewed in Section
\ref{sec:2.4}.
Some further results obtained by related arguments are mentioned in Section
\ref{sec:2.5}.  
\item
In Section \ref{sec:3} we discuss the Lundholm--Solovej method.
This approach is based on a suitable decomposition of the Laplacian in the
position space.
The key observation is that Lieb--Thirring inequalities can be deduced from a
local version of the exclusion principle.
Interestingly, such a local exclusion bound was used also by Dyson and Lenard
in their first proof of the stability of matter \cite{DysLen-67}.
The approach in \cite{LunSol-13} was originally proposed to study anyons
(particles with only fractional statistics), but it has led to several new
Lieb--Thirring inequalities.
We will explain the general strategy in Section \ref{sec:3.1} and review some
new results in the next subsections.  
\end{itemize}
Note that there are other direct approaches which are not covered here, see
e.g. Sabin's work \cite{Sabin-16}.
Moreover, although the methods represented in this note seem fascinating,
there are various extensions of Lieb--Thirring inequalities that have to be
treated differently.
We refer to Frank's review \cite{Frank-20} for further aspects of the subject. 

To end the introduction, let us mention that the methods covered here are
general enough to handle fractional Lieb--Thirring inequalities with rather
little effort.
Therefore, we will in most cases aim at this generalization.
In particular, a more general version of Theorem \ref{thm:LT-kinetic-intro} is

\begin{theorem}\mbox{\normalfont(Fractional Lieb--Thirring kinetic inequality)}
\label{thm:intro-LT-kin-frac}
Let $d\ge 1$ and $s>0$.
For any orthonormal functions $\{u_n\}_{n=1}^N$  in $L^2(\R^d)$ with density
$\rho(x)=\sum_{n=1}^N |u_n(x)|^2$, we have 
\begin{align} \label{eq:LT-ineq-frac} 
  \sum_{n=1}^N \bigl\| (-\Delta)^{s/2} u_n \bigr\|_{L^2(\R^d)}^2
  \ge K_{d,s} \int_{\R^d} \rho(x)^{1+\frac {2s}{d}}\,\d x. 
\end{align}
The constant $K_{d,s}>0$ is independent of $N$ and $\{u_n\}_{n=1}^N$. 
\end{theorem}
\textbf{Notation.} 
We will often denote by $C$ a general positive constant whose value may change
from line to line.
In some cases, the dependence on key parameters will be included in
the notation.
However, it is important that all constants are always independent of the number of
variables $N$.
Sometimes we write $\int_{\Omega}f$ instead of $\int_{\Omega}f(x)\,\d x$.

\section{Rumin method} \label{sec:2}

\subsection{A simple proof of Sobolev inequality} \label{sec:2.1}
Recall the definition of the Sobolev space  
\[ H^s(\mathbb{R}^d):= \left\{ u \in L^2(\mathbb{R}^d) \,\big| \, |k|^s
    \widehat u(k) \in L^2(\mathbb{R}^d) \right\} \]
with an arbitrary power $s\ge 0$ (not necessarily an integer).
Here we use the following convention of the Fourier transform \cite{LieLos-01} 
$$
\widehat u(k)= \int_{\R^d} e^{-2\pi \I k \cdot x} u(x) \d x
$$
with $\I$ the imaginary unit (${\I^2=-1}$).
Thus $H^s(\mathbb{R}^d)$ is a Hilbert space with the inner product
\[ \langle{u,v}\rangle_{H^s(\R^d)}
  = \int\limits_{\mathbb{R}^d} \overline{\widehat u(k)}\,
  \widehat v(k) \bigl(1+|2\pi k|^{2s}\bigr)\, \d k.  \]
On $H^s(\R^d)$, the weak derivatives can be defined via the Fourier transform
$$
\widehat{D^\alpha u} (k) =  (2\pi \I k)^{\alpha}  \widehat {u} (k) \in L^2(\R^d)
$$
for any ${\alpha=(\alpha_1,\dots,\alpha_d)\in \{0,1,\dots\}^d}$ with
${|\alpha|=\alpha_1+\cdots+\alpha_d \le s}$.
In particular,
\begin{eqnarray*}
  \bigl\langle u, (-\Delta)^{s} u \bigr\rangle_{L^2(\R^d)}
  &=& \bigl\| (-\Delta)^{s/2} u \bigr\|_{L^2(\R^d)}^2 \\
  &=& \int_{\R^d} |2\pi k|^{2s} \bigl|\widehat u(k)\bigr|^2\, \d k,
      \quad \forall u\in H^s(\R^d).   
\end{eqnarray*}

As a warm-up, let us consider 
\begin{theorem}\mbox{\normalfont(Standard Sobolev inequality)}
If $d >2s \ge 0$, then  
$$
\|f\|_{L^{\frac{2d}{d-2s} }(\R^d)}
\le C_{d,s} \bigl\| (-\Delta)^{s/2} f\bigr\|_{L^2(\R^d)},
\quad \forall f\in H^s(\R^d). 
$$
\end{theorem}
\begin{proof}
The following proof is due to Chemin and Xu \cite{CheXu-97}.
We use the momentum decomposition based on the identity
$$
|2\pi k|^{2s} = \int_0^{\infty} \1\bigl(|2\pi k|^{2s}>E\bigr)\,\d E.  
$$
For every $u\in H^s(\R^d)$, by Plancherel's and Fubini's theorems we can write
\begin{align} \label{eq:K-uE+}
  K &:=\| (-\Delta)^{s/2} u \|_{L^2(\R^d)}^2
      = \int_{\R^3} |2\pi k|^{2s} |\widehat u (k)|^2\,\d k \nn\\
    &=   \int_{\R^3}\left(\int_0^\infty \1\bigl(|2\pi k|^{2s}>E\bigr)
      |\widehat u(k)|^2\,\d E\right)\d k\nn\\
    &=   \int_0^\infty\left(\int_{\R^3} \1\bigl(|2\pi k|^{2s}>E\bigr)
      |\widehat u(k)|^2\,\d k\right)\d E \nn\\
    &=   \int_0^\infty\left(\int_{\R^3} 
      |\widehat u^{E_+}(k)|^2\,\d k\right)\d E 
       = \int_{\R^3} \left(  \int_0^\infty |u^{E+}(x)|^2\,\d E\right)\d x 
\end{align}
where  the function $u^{E+}$ is defined via the Fourier transform
$$\widehat u^{E+}(k) = \1\bigl(|2\pi k|^{2s}>E\bigr) \widehat u(k).$$
When $d>2s$, by H\"older's inequality we have the uniform bound
\begin{align*}
  \bigl|u(x)- u^{E+}(x)\bigr|
   & = \biggl| \int_{\R^d}  e^{2\pi\I k \cdot x}
    \Bigl( \widehat{u}(k) - \widehat{ u}^{E+}(k)\Bigr)\,\d k\biggr|\\
   &=\biggl| \int_{\R^d} e^{2\pi\I k  \cdot x}
     \1\bigl(|2\pi k|^{2s} \le E\bigr) \widehat{u}(k)\,\d k \biggr| \\
& \le \biggl( \int_{\R^d}|2\pi k|^{2s} |\widehat{u}(k)|^2\,\d k\biggr)^{\!1/2}
  \biggl( \int_{\R^d}\frac{\1\bigl(|2\pi k|^{2s} \le E\bigr)}
          {|2\pi k|^{-s}}\,\d k\biggr)^{\!1/2}
     \\&= C \sqrt{K}  E^{\frac{d-2s}{4s}}
\end{align*}
with a finite constant $C=C(d,s)>0$.
By the triangle inequality, 
$$
\bigl|u^{E+}(x)\bigr| \ge \Bigl| \bigl|u(x)\bigr|
- \bigl|u(x)- u^{E+}(x)\bigr| \Bigr|
\ge \Bigl[ \bigl|u(x)\bigr| - C \sqrt{K} E^{\frac{d-2s}{4s}} \Bigr]_+
$$
where $t_+=\max\{t,0\}$ is the positive part of $t$.
Integrating over $E$ we get 
\begin{align*}
  \int_0^\infty \bigl|u^{E+}(x)\bigr|^2\,\d E
  &\ge \int_0^\infty \Bigl[ \bigl|u(x)\bigr|
    - C \sqrt{K} E^{\frac{d}{2s}-1} \Bigr]_+^2\,\d E \\
  &\ge C \bigl|u(x)\bigr|^{\frac{2d}{d-2s}}  K^{- \frac{2s}{d-2s}}.
\end{align*}
Inserting the latter bound in \eqref{eq:K-uE+}, we arrive at 
$$
K \ge C   K^{- \frac{2s}{d-2s}} \int_{\R^d}\bigl|u(x)\bigr|^{\frac{2d}{d-2s}}\,\d x,
$$
which is equivalent to the desired inequality.
\end{proof}

\subsection{Lieb--Thirring kinetic inequality} \label{sec:2.2}
The previous proof can be extended to bound the kinetic energy of a
family of orthonormal functions.  
\begin{theorem}\mbox{\normalfont(Lieb--Thirring kinetic inequality)}
Let ${d>2\kappa\ge 0}$ and ${s\ge 0}$
($\kappa$ and $s$ are not necessarily integers).
For any ${N\ge 1}$, let $\{(-\Delta)^{\kappa/2} u_n\}_{n=1}^N$ be orthonormal
functions in $L^2(\R^d)$ and denote $\rho(x)=\sum_{n=1}^N |u_n(x)|^2$.
Then
\bq \label{eq:LT-kinetic-s-d-kappa}
\sum_{n=1}^N \| (-\Delta)^{s/2} u_n \|_{L^2(\R^d)}^2
\ge K_{d,s,\kappa} \int_{\R^d}  \rho(x)^{1+\frac {2s}{d-2\kappa}} \,\d x.
\eq
The constant $K_{d,s,\kappa}$ is independent of $N$ and $\{u_n\}_{n=1}^N$. 
\end{theorem}
The case $\kappa=0$ is Theorem \ref{thm:intro-LT-kin-frac} in the
Introduction.
The case $\kappa=s$ is Lieb's inequality \cite{Lieb-83b}, which implies the 
Cwikel--Lieb--Rozenblum inequality as we will see.
The following proof for ${\kappa=s}$ is due to Rumin \cite{Rumin-10} (see also
Frank \cite{Frank-14}).
The result for ${\kappa\not\in \{0,s\}}$ appears here for the first time.

\begin{proof} By Plancherel's and Fubini's theorems we can write
\begin{align} \label{eq:K-un-E+}
  \sum_{n=1}^N \| (-\Delta)^{s/2} u_n \|_{L^2(\R^d)}^2
  = \int_{\R^3} \left(\int_0^\infty \sum_{n=1}^N \bigl|u_n^{E+}(x)\bigr|^2
  \,\d E\right) \d x
\end{align}
where the function $u_n^{E+}$ is defined via the Fourier transform
$$\widehat u_n^{E+}(k) = \1\bigl(|2\pi k|^{2s}>E\bigr) \widehat u_n(k).$$
Now we use the assumption that $\{(-\Delta)^{\kappa/2} u_n\}_{n=1}^N$ are
orthonormal functions in $L^2(\R^d)$.
This implies that $\{ |2\pi k|^{\kappa} \widehat u_n(k)\}_{n=1}^N$ are
orthonormal functions in $L^2(\R^d, \d k)$.
Hence, by Bessel's inequality we have the uniform bound 
\begin{align*}
  \sum_{n=1}^N \bigl|u_n(x)- u_n^{E+}(x)\bigr|^2
  &= \sum_{n=1}^N \biggl| \int_{\R^d} 
    e^{2\pi\I k  \cdot x} \1\bigl(|2\pi k|^{2s} \le E\bigr)
    \widehat{u}_n(k)\,\d k\biggr|^2 \\
  &= \sum_{n=1}^N \biggl| \int_{\R^d}
    e^{2\pi\I k  \cdot x} \frac{\1\bigl(|2\pi k|^{2s} \le E\bigr)}
    {|2\pi k|^{\kappa} }
    |2\pi k|^{\kappa}  \widehat{u}_n(k)\,\d k \biggr| ^2\\
  & \le \biggl\| e^{2\pi\I k  \cdot x}
    \frac{\1\bigl(|2\pi k|^{2s} \le E\bigr)}{|2\pi k|^{\kappa} }
    \biggr\|_{L^2(\R^d, \d k)}^2 = C E^{\frac{d-2\kappa }{2s}}.
\end{align*}
Here the constant $C=C(d,\kappa)>0$ is finite when $d>2\kappa$. 

Next, by the triangle inequality for vectors in $\mathbb{C}^N$, we have 
\begin{align*}
  \biggl( \sum_{n=1}^N \bigl|u_n^{E+}(x)\bigr|^2 \biggr)^{\frac 1 2}
  &\ge \Biggl| \biggl( \sum_{n=1}^N \bigl|u_n(x)\bigr|^2 \biggr)^{\frac 1 2}
    - \biggl( \sum_{n=1}^N \bigl|u_n(x)-u_n^{E+}(x)\bigr|^2
    \biggr)^{\frac 1 2} \Biggr|
  \\&\ge \Bigl[ \sqrt{\rho(x)} - C E^{\frac{d-2\kappa}{4s}} \Bigr]_+.  
\end{align*}
Consequently,
\begin{align*}
  \int_0^\infty \sum_{n=1}^N \bigl|u_n^{E+}(x)\bigr|^2\,\d E
  &\ge \int_0^\infty  \biggl(\Bigl[ \sqrt{\rho(x)}
    - C E^{\frac{d-2\kappa}{4s}} \Bigr]_+ \biggr)^2\,\d E
  \\&\ge C  \rho(x)^{1+\frac {2s}{d-2\kappa}}.  
\end{align*}
Inserting the latter bound in \eqref{eq:K-un-E+}, we obtain the desired
inequality \eqref{eq:LT-kinetic-s-d-kappa}.\rule{1em}{0.0pt}  
\end{proof}

The extension of the above result to the case $d\le 2\kappa$ requires some
modification.
Here let us focus only on the case $\kappa=s$. 
\begin{theorem}\mbox{\normalfont(Kinetic inequality in low dimensions)}
\label{thm:LT-low}Let $2s\ge d\ge 1$.
For any $N\ge 1$ and $E>0$, let
$\{\sqrt{(-\Delta)^{s}+E} \,u_n\}_{n=1}^N$ be orthonormal functions in
$L^2(\R^d)$ and denote $\rho(x)=\sum_{n=1}^N |u_n(x)|^2$.  
\begin{itemize}
\item If $d=2s$, then there exist constants $C_d,\alpha_d>0$ such that 
$$
N \ge  E \int_{\R^d}  f\bigl(\rho(x)\bigr)\,  \d x,
\quad f(t)=C_{d} te^{t\alpha_d}.
$$
\item If $d<2s$, then $\rho(x) \le C_{d,s} E^{\frac d {2s}-1}$
  for a.e.\ $x\in \R^d$. 
\end{itemize}
\end{theorem}

\begin{proof} By Plancherel's and Fubini's theorems we can write
\begin{align} \label{eq:K-un-L+}
  N=\sum_{n=1}^N \left\| \sqrt{(-\Delta)^{s}+E}\, u_n \right\|_{L^2(\R^d)}^2
  = \int_{\R^3} \left(\int_0^\infty \sum_{n=1}^N \bigl|u_n^{L+}(x)\bigr|^2
  \,\d L\right)\d x
\end{align}
where  the function $u_n^{L+}$ is defined via the Fourier transform
$$\widehat{u}_n^{L+}(k) = \1\bigl(|2\pi k|^{2s}+E>L\bigr)\widehat{u}_n(k).$$
Now we use the assumption that $\{\sqrt{(-\Delta)^{s}+E}\, u_n\}_{n=1}^N$ are
orthonormal functions in $L^2(\R^d)$.
This implies that $\{ \sqrt{ |2\pi k|^{2s} +E} \,\widehat u_n(k)\}_{n=1}^N$
are orthonormal functions in $L^2(\R^d, \d k)$.
Hence, by Bessel's inequality we have the uniform bound 
\begin{align*}
  \sum_{n=1}^N \bigl|u_n(x)- u_n^{E+}(x)\bigr|^2
  &= \sum_{n=1}^N \biggl| \int_{\R^d} \d k \,
    e^{2\pi\I k \cdot x} \1\bigl(|2\pi k|^{2s}  +E \le L\bigr)
    \widehat{u}_n (k)\,\d k\biggr|^2 \\
  &\le \int_{\R^d} \frac{\1\left(|2\pi k|^{2s} +E \le L\right)}
    { |2\pi k|^{2s} +E}\,\d k
  \\&\le C_d E^{\frac d {2s}-1} \int_{0}^{(\frac L E-1)_
  +^{\frac 1 {2s}}} \frac {r^{d-1} }{ r^{2s} +1}\,\d r.
\end{align*}

{\bf Case $d=2s$.} When $L\ge E$ we get 
\begin{align*}
  \sum_{n=1}^N \bigl|u_n(x)- u_n^{E+}(x)\bigr|^2
  &\le C_d  \int_{0}^{(\frac L E-1)_+^{\frac 1 {d}}}
    \frac {r^{d-1} }{ r^{d} +1}\,\d r \le C_d  \log (L/E).
\end{align*}
Hence, by the triangle inequality 
\begin{align*}
  \int_E^\infty \sum_{n=1}^N \bigl|u_n^{L+}(x)\bigr|^2\,\d L
  &\ge \int_E^\infty \left( \Bigl[ \sqrt{\rho(x)}
    - \sqrt{C_d \log (L/E) } \Bigr]_+ \right) ^2\,\d L\\
  &= E \int_1^\infty \left( \Bigl[ \sqrt{\rho(x)}
    - \sqrt{C_d \log (L) } \Bigr]_+ \right)^2\, \d L
  \\&\ge E  \rho(x) e^{\alpha_d \rho(x)}. 
\end{align*}
Inserting the latter bound in \eqref{eq:K-un-L+}, we obtain the desired
inequality. 

{\bf Case $d<2s$.} When $L\ge E$ we get 
\begin{align*}
  \sum_{n=1}^N \bigl|u_n(x)- u_n^{E+}(x)\bigr|^2
  &\le C_d E^{\frac d {2s}-1} \int_{0}^{\infty}
    \frac {r^{d-1}}{ r^{2s} +1}\,\d r \le C_d E^{\frac d {2s}-1}. 
\end{align*}
Hence, if $\rho(x) > C_d E^{\frac d {2s}-1}$, then by the triangle inequality
\begin{align*}
  \int_E^\infty \sum_{n=1}^N \bigl|u_n^{L+}(x)\bigr|^2\,\d L
  &\ge \int_E^\infty \left( \Bigl[ \sqrt{\rho(x)}
    - \sqrt{C_d E^{\frac d {2s}-1}} \Bigr]_+ \right)^2\,\d L= \infty. 
\end{align*}
Thus from \eqref{eq:K-un-L+} we conclude that
$\rho(x) \le C_d E^{\frac d {2s}-1}$ for a.e.\ $x\in \R^d$. 
\end{proof}

Note that in the above result, the case ${d=2s}$ is related to the question
discussed in \cite[Section 5.8]{Frank-20}.
Moreover, in the case ${d<2s}$, instead of using Rumin's method we can also
proceed as follows.
Denote $v_n=\sqrt{(-\Delta)^s+E}\,u_n$, then 
$$u_n=\bigl((-\Delta)^s+E\bigr)^{-1/2}v_n
= G_E*v_n, \quad  \widehat{G}_E(k)= \bigl(|2\pi k|^{2s}+E\bigr)^{-1/2}.$$
Since $\{v_n\}_{n=1}^N$ are orthonormal functions in $L^2(\R^d)$, by Bessel's
inequality  and Plancherel's theorem we have the uniform bound 
\begin{align*}
  \rho(x)=\sum_{n=1}^N \biggl|u_n(x)\biggr|^2
  = \sum_{n=1}^N \biggl| \int_{\R^d} G(x-y) v_n(y)\, \d y\biggr|^2
  \le \|G\|_{L^2}^2 = C_d E^{\frac d {2s}-1}. 
\end{align*}

\subsection{Eigenvalue bounds for Schr\"odinger operators} \label{sec:2.3}
Now let us consider an extension of Theorem \ref{thm0:LT} for the fractional
Laplacian.  
\begin{theorem}\mbox{\normalfont(Eigenvalue bounds for fractional Laplacian)}
\label{thm:LT-ev-frac}
Let ${d \ge 1}$, ${s>0}$ and ${\kappa \ge 0}$.
Let ${V:\R^d \to \R}$ be a real-valued potential such that
${V_- \in L^{\kappa+\frac d {2s}}(\R^d)}$.
Assume that the Schr\"odinger operator $(-\Delta)^s+V(x)$ on $L^2(\R^d)$ has
negative eigenvalues $\{E_n((-\Delta)^s+V)\}_{n\ge 1}$.
Then 
\bq \label{eq:LT-ineq}
\sum_{n\ge 1} \bigl|E_n\bigl((-\Delta)^s+V\bigr)\bigr|^\kappa
\le L_{\kappa,d,s} \int_{\R^d} |V(x)_-|^{\kappa+\frac d {2s}}\,\d x
\eq
for a finite constant $L_{\kappa,d,s}$ independent of $V$, provided that
\[\begin{cases}
   \kappa \ge 0                  &\mathrm{if}\ d> 2s,\\
   \kappa >0                     &\mathrm{if}\ d=2s,\\
   \kappa \ge 1- \frac {d} {2s}  &\mathrm{if}\ d<2s.\\
  \end{cases}	
\]
\end{theorem}
The Lieb--Thirring inequality for a generalized kinetic operator $f(\I\nabla)$
was studied in \cite{Daubechies-83}.
In the form of \eqref{eq:LT-ineq}, all non-critical cases can be treated using
the method in \cite{LieThi-76}.
The critical case $\kappa= 0$ follows from the proofs in
\cite{Cwikel-77,Rozenblum-76}.
For the other critical case $\kappa= 1- d/(2s)$ see
\cite{Weidl-96,NetWei-96,Frank-18}.

In the following we prove Theorem \ref{thm:LT-ev-frac} using the kinetic
inequalities in Section \ref{sec:2.2}.
The proof covers all cases except the critical case $\kappa= 1-d/(2s)$. 
\begin{proof}
{\bf Case 1:} $d>2s$, $\kappa=0$.
The following duality argument is due to Frank \cite{Frank-14}.
Let $W$ be the space spanned by eigenfunctions of negative eigenvalues of
$(-\Delta)^s+V$.
Assume that $\dim W \ge N$.
Since the operator $(-\Delta)^{s/2}$ is strictly positive on $L^2(\R^d)$, we
get  
$$\dim\bigl((-\Delta)^{s/2}W\bigr)\ge N.$$ 
Thus we can choose $\{u_n\}_{n=1}^N \subset W$ such that
$\{(-\Delta)^{s/2}u_n\}_{n=1}^N$ are orthonormal in $L^2(\R^d)$.
By the kinetic inequality \eqref{eq:LT-kinetic-s-d-kappa} with $\kappa=s$, we
have 
\begin{eqnarray}\label{eq:N-rho-standard}
  &&N = \sum_{n=1}^N \bigl\|(-\Delta)^{s/2}u_n\bigr\|_{L^2(\R^d)}^2
  \ge K_{d,s} \int_{\R^d} \rho^{\frac d {d-2s}}(x)\,\d x\nn\\
&&\mbox{with}\quad\rho(x):=\sum_{n =1}^N \bigl| u_n(x)\bigr|^2.  
\end{eqnarray}
On the other hand, since $\{u_n\}_{n=1}^N \subset W$ we have
\begin{align*}
  0 \ge \sum_{n =1}^N
  \bigl\langle u_n, \bigl((-\Delta)^s+V\bigr)u_n \bigr\rangle_{L^2(\R^d)}
  = N + \int_{\R^d} V(x) \rho(x)\,\d x. 
\end{align*}
Putting together these inequalities, we find that
\begin{eqnarray*}
  N &\le& - \int_{\R^d} V(x)\rho(x)\,\d x
          \le \int_{\R^d} |V_-(x)|\rho(x)\,\d x\\
&\le& \|V_-\|_{L^{\frac d {2s}}} \|\rho\|_{L^{\frac {d}{d-2s}}}
\le \|V_-\|_{L^{\frac d {2s}}} \Big( \frac{N}{K_{d,s}}\Big)^{\frac {d-2s}{d}},  
\end{eqnarray*}
which is equivalent to 
$$
N \le K_{d,s}^{1- \frac d {2s}} \int_{\R^d} |V_-(x)|^{\frac d {2s}} \,\d x.
$$
Then we conclude by taking $N \to \dim W$.
	
\medskip\noindent%
{\bf Case 2:} $d>2s$, $\kappa>0$.
The result for ${\kappa>0}$ follows from the case ${\kappa=0}$, thanks to a
general argument of Aizenman and Lieb \cite{AizLie-78}.
We use the layer-cake representation 
$$
|E_n|^\kappa = \kappa \int_0^\infty
\1\bigl( E_n < - E\bigr) E^{\kappa-1}\, \d E.
$$
Using the bound in Case 1 for the number of negative eigenvalues of
$(-\Delta)^s +V +E$, we have 
$$
\sum_{n\ge 1}\1( E_n < - E)
\le C \int_{\R^d} \bigl|(V(x)+E)_-\bigr|^{\frac d {2s}}\, \d x
$$
for a constant $C=C(d,s)>0$.
Thus 
\begin{align*}
  \sum_{n\ge 1} |E_n|^\kappa
  &=  \kappa \int_0^\infty  \sum_{n\ge 1}\1( E_n < - E) E^{s-1} \,\d E \\
  &\le C \kappa\int_0^\infty \left(\int_{\R^d}
    \bigl|\bigl(V(x)+E\bigr)_-\bigr|^{\frac d {2s}}
    E^{\kappa -1}\,\d x\right)\d E \\
  &\le C \kappa \int_{\R^d} \left(  \int_0^\infty 
    \bigl|\bigl(V(x)+E\bigr)_-\bigr|^{\frac d {2s}}
    E^{\kappa -1}\,\d E\right)\d x \\
  &= C' \int_{\R^d} \bigl|V(x)_-\bigr|^{\kappa + \frac d {2s}}\,\d x,
\end{align*}
where
$$
C' = C \kappa \int_0^\infty  \bigl|(1+E)_-\bigr|^{\frac d {2s}}
E^{\kappa -1}\,\d E <\infty. 
$$

\medskip\noindent%
{\bf Case 3:} $d=2s$, $\kappa>0$.
Let $N_E$ be the number of negative eigenvalues of $(-\Delta)^s +V +E$.
We can bound $N_E$ by arguing as in Case 1, but replacing
\eqref{eq:N-rho-standard} by Theorem \ref{thm:LT-low}, namely 
\begin{align*}
  N_E \ge  E \int\limits_{\mathbb{R}^2} f\bigl(\rho(x)\bigr)\,\d x,
  \quad f(t) =C_d te^{t\alpha_d}  
\end{align*}	
with constants $C_d,\alpha_d>0$.
Therefore,  
\begin{align} \label{eq:Ne-pre}
  N_E \le 2 \int_{\R^d} \bigl|V_-(x)\bigr| \rho(x)\,\d x - N_E
  &\le \int_{\R^d } \Bigl( 2\bigl|V_-(x)\bigr|\rho(x)
  - E f\bigl(\rho(x)\bigr) \Bigr)\,\d x \nn\\
  &\le    \int_{\R^d}  E f^*{\left( \frac{2\bigl|V_-(x)\bigr|}{E}\right)}
  \,\d x
\end{align}
where $f^*:[0,\infty) \to  [0,\infty]$ is the Legendre transform of $f$,
namely 
$$
f^*(y)= \sup_{t\ge 0} \{yt- f(t)\}, \quad \forall y\ge 0.
$$
Since $V+E=(V+E/2) +E/2$, we can also replace $(V,E)$ by $(V+E/2,E/2)$ and
deduce from \eqref{eq:Ne-pre} that  
\begin{align*}
  N_E \le  \frac E 2 \int_{\R^d}
  f^*{\left(\frac{4\bigl|\bigl(V(x) +E\bigr)_-\bigr|}{E}\right)}\,\d x. 
\end{align*}
Then we argue as in case 2 and obtain, by the layer-cake representation, 
\begin{align*}
  \sum_{n\ge 1} |E_n|^\kappa
  &= \kappa \int_0^\infty  N_E E^{\kappa-1} \d E\\
  &\le \kappa   \int_0^\infty  \left(\int_{\R^d}
    f^*{\left( \frac{2\bigl|(V(x)+E/2)_-\bigr|}{(E/2)}\right)}
    \frac{E^{\kappa}}{2}\,\d x\right)\d E \\
  &= C_{d,\kappa}
    \left(\int_{\R^d} \bigl|V(x)_-\bigr|^{\kappa+1}\,\d x\right)
    \left(\int_{0}^{1} f^*{\left(\frac{2}{y} - 2\right)} y^\kappa\,\d y\right).
\end{align*}
In the last equality we have changed the variable $E=2\bigl|V(x)_-\bigr|y$.
It remains to show that: 
$$
\int_{0}^{1} f^*{\left( \frac{2}{y} - 2\right)}
y^{\kappa}\,\d y<\infty.
$$
Note that if $f\ge g$, then $f^*\le g^*$.
Moreover, $(t^p/p)^* = t^q/q$ with $1/p+1/q=1$ by Young's inequality.
Since $f(t)$ grows faster than any polynomial, we find that 
$$f^*(t)\le C_q t^q$$
for any $q\in (1,1+\kappa)$.
Hence
$$
\int_{0}^{1} f^*{\left( \frac{2}{y} - 2\right)} y^{\kappa}\,\d y
\le C_q \int_{0}^{1} {\left( \frac{2}{y} - 2\right)}^q y^{\kappa}\,\d y
\le  C_q \int_{0}^{1} y^{\kappa-q}\,\d y <\infty. 
$$

\medskip\noindent%
{\bf Case 4:} $d<2s$, $\kappa> 1- \frac d {2s}$.
Again we bound $N_E$, the number of negative eigenvalues of $(-\Delta)^s +V
+E$, as in case 1 but replacing \eqref{eq:N-rho-standard} by the uniform bound 
$$
\rho(x) \le C_d E^{\frac d {2s}-1}, \quad \text{ a.e. }x\in \R^d,
$$
from Theorem \ref{thm:LT-low}. Thus 
$$
N_E \le \int_{\R^d} |V_-|\rho
\le C_d E^{\frac d {2s}-1} \int_{\R^d} \bigl|V_-(x)\bigr|\,\d x.
$$
Again, we can write  $V +E= (V +E/2) + E/2$ and obtain
$$
N_E \le C_d (E/2)^{\frac d {2s}-1}
\int_{\R^d} \bigl|\bigl(V(x)+E/2\bigr)_-\bigr|\,\d x.
$$ 
Thus by the layer-cake representation, 
\begin{align*}
  \sum_{n\ge 1} |E_n|^\kappa
  &= \kappa \int_0^\infty  N_E E^{\kappa-1} \,\d E\\
  &\le \kappa C_d  \int_0^\infty  E^{\kappa-1}(E/2)^{\frac d {2s}-1}
    \left(\int_{\R^d} \bigl|(V(x)+E/2)_-\bigr|\,\d x\right)\d E\\
  &= C_{d, \kappa,s}
    \left(\int_{\R^d} \bigl|V_-(x)\bigr|^{\kappa+ \frac d {2s}}\, \d x\right)
    \left(\int_0^\infty  E^{ \kappa + \frac d {2s}-2}
    (1-E)_+\,\d E\right).
\end{align*}
Finally, when $\kappa> 1- d/(2s)$ we have 
$$
\int_0^\infty   E^{\kappa+ \frac d {2s}-2} (1-E)_+\,\d E
=\int_0^1 E^{\kappa+ \frac d {2s}-2} (1-E)\,\d E <\infty. 
$$
This completes the proof of Theorem \ref{thm:LT-ev-frac}
(except for the critical case ${\kappa= 1- d/(2s)}$).  
\end{proof}

\subsection{Best known constant for kinetic inequality} \label{sec:2.4}
In the non-relativistic case, Lieb and Thirring \cite{LieThi-75,LieThi-76}
conjectured that the optimal constant in the kinetic inequality
\eqref{eq:intro-LT-ineq-kin} is
\bq \label{eq:LT-conj-K}
K_{d}=\min\{K_{d}^{\rm cl}, C_{\mathrm{GN}}(d)\}=
\begin{cases}
K_{d}^{\rm cl} & \text{ if } d\ge 3, \\
C_{\rm GN}(d) & \text{ if } d=1,2,  
\end{cases}
\eq
where $C_{\rm GN}(d)$ is defined in \eqref{eq:C-GN}.
Thanks to the relation \eqref{eq:L-K}, the Lieb--Thirring conjecture is
equivalent to  
\bq \label{eq:LT-conj}
L_{1,d}=\max\{L_{1,d}^{\rm cl}, L_{1,d}^{\rm So}\}=
\begin{cases}
L_{1,d}^{\rm cl} & \text{ if } d\ge 3, \\
L_{1,d}^{\rm So} & \text{ if } d=1,2,  
\end{cases}
\eq
where $L_{1,d}^{\rm So}$ is the optimal constant in the one-body bound 
\bq \label{eq:So}
\int_{\R^d}\Bigl( \bigl|\nabla u(x)\bigr|^2
+ V(x)\bigl|u(x)\bigr|^2 \Bigr)\,\d x
\ge - L_{1,d}^{\rm So} \int_{\R^d} \bigl|V(x)_-\bigr|^{1+d/2}\,\d x. 
\eq
The original proof of Lieb and Thirring \cite{LieThi-75} gave
$L_{1,d}/L_{1,d}^{\rm cl}\le 4\pi$ in ${d=3}$.
This bound has been improved further in
\cite{Lieb-84,EdeFoi-91,BlaStu-96,HunLapWei-00,DolLapLos-08}.
The latest improvement in \cite{FraHunJexNam-18} is 
\begin{theorem} \label{thm:best-kinetic}
  For all $d\ge 1$ we have $ (K_{d}^{\rm cl}/K_d)^{d/2}
  =L_{1,d}/L_{1,d}^{\rm cl} \le  1.456.$
\end{theorem}
In one dimension, this bound is about $26\%$ bigger than the expected value
$L_{1,1}^{\rm So}/L_{1,1}^{\rm cl} = 2/\sqrt{3}=1.155\dots$ in
\cite{LieThi-76}.
For any $d\ge 1$, by the Aizenman--Lieb monotonicity \cite{AizLie-78}, Theorem
\ref{thm:best-kinetic} implies that $L_{\kappa,d}/L_{\kappa,d}^{\rm cl} \le
1.456$ for all $\kappa\ge 1$, which is also the best known result for all
$1\le \kappa<3/2$ (when $\kappa\ge 3/2$, we know that
$L_{\kappa,d}=L_{\kappa,d}^{\rm cl}$ \cite{LieThi-76,LapWei-00}).  

The proof of Theorem \ref{thm:best-kinetic} uses crucially the technique of
optimizing momentum decompositions.
More precisely, it contains two main steps. 
\begin{itemize}
\item First, we improve the kinetic inequality using a modification of Rumin's
  method.
  This gives $L_{1,d}/L_{1,d}^{\rm cl}  \le 1.456$ in $d=1$ (and worse bounds
  in higher dimensions).  

\item Second, we use the Laptev--Weidl lifting argument \cite{LapWei-00} to
  extend the bound to higher dimensions.  
\end{itemize}
The first step can be extended to fractional Laplacian to bound the constant
in Theorem \ref{thm:intro-LT-kin-frac}.
For every $s>0$, the corresponding semiclassical  constant is 
$$ K_{d,s}^{\rm cl}
=  \frac{d}{d+2s} \left( \frac{(2\pi)^d}{|B_1|} \right)^{\frac{2s}{d}}.$$
We have
\begin{theorem} \label{thm:fractional}
For all $d\ge 1$ and $s>0$, the best constant in \eqref{eq:LT-ineq-frac}
satisfies 
\begin{align*} 
  \frac{K_{d,s}}{K_{d,s}^{\rm cl}} \ge   \frac{d}{d+2s}
  {\left( \frac{2s}{d+2s}\right)}^{\frac{4s}{d}}
  \mathcal{C}_{d,s}^{-\frac{2s}{d}} 
\end{align*}
where
\begin{align} \label{eq:inf-Cfl-frac}
  \mathcal{C}_{d,s}
  &:= \inf \left\{
    {\left( \int_0^\infty\!\!\varphi(r)^2\,\d r\right)}^{\frac{d}{2s}}
    \frac{d}{2s}\int_0^\infty\! 
    {\left| 1-\int_0^\infty \!\varphi(E) f(Et)
    \,\d E\right|}^2 \frac{\d t}{t^{1+\frac{d}{2s}}}  \right\} 
\end{align}
with the infimum taken over all functions $f,\varphi:[0,\infty) \to
[0,\infty)$ satisfying $\int_0^\infty f(r)^2\,\d r=1$.  
\end{theorem}
When $d=s=1$, we have $\mathcal{C}_{1,1}\le 0.373556$ by taking in
\eqref{eq:inf-Cfl-frac}  
$$
f(t)= (1+\mu_0 t^{4.5})^{-0.25},  \quad
\varphi(t)= \frac{(1-t^{0.36})^{2.1}}{1+t}\1(t\le 1)
$$
with $\mu_0$ determined by $\int_0^\infty f^2= 1$.
This implies $L_{1,1}/L_{1,1}^{\rm cl} \le 1.456.$

\begin{proof} Using $\int_0^\infty f(r)^2\,\d r=1$ we can write 
$$
|2\pi k|^{2s}= \int_0^\infty f(E |2\pi k|^{-2s})^2\,\d E.
$$ 
Thus  
\begin{align}\label{eq:kinetic-representationfrac}
  \sum_{n=1}^N \bigl\| (-\Delta)^{s/2} u_n\bigr\|_{L^2}^2
  = \int_{\R^d} \left( \sum_{n=1}^N \int_0^\infty
  \bigl|u_n^{E+}(x)\bigr|^2\, \d E \right) \d x
\end{align}
where
$$
\widehat{u}_n^{E+}(k) = f\bigl(E |2\pi k|^{-2s}\bigr) \widehat{u}_n(k). 
$$
Next, for every function $\varphi:[0,\infty) \to [0,\infty)$, by the
Cauchy--Schwarz inequality and the triangle inequality in $\mathbb{C}^d$, we
have   
\begin{align*} 
  & \sum_{n=1}^N \left( \int_0^\infty\! \varphi(E)^2\, \d E \right)
    \left(  \int_0^\infty \bigl|u_n^{E+}(x)\bigr|^2\, \d E  \right)
    \ge\sum_{n=1}^N \left|\int_0^\infty \varphi(E) u_n^{E+}(x)\, \d E \right|^2\\
  &\ge \left|  \left( \sum_{n=1}^N  \bigl|u_n(x)\bigr|^2 \right)^{\!1/2} -
    \left( \sum_{n=1}^N  \Bigl| u_n(x)
    - \int_0^\infty\! \varphi(Et) u_n^{E+}(x)\, \d E \Bigr|^2
    \right)^{\!1/2} \right|^2.
\end{align*}
Next, using again the fact that $\{u_n\}_{n=1}^N$ are orthonormal functions in
$L^2(\R^d)$ and Bessel's inequality, we have the uniform bound
\begin{align*} 
  &\;\sum_{n=1}^N  \left| u_n(x)
    - \int_0^\infty \varphi(E) u_n^{E+}(x)\, \d E \right|^2 \\
  =&\;\sum_{n=1}^N \left| \int_{\R^d} e^{2\pi\I k \cdot x}
    \widehat u_n(k) \left( 1- \int_0^\infty \varphi(E)
    f(E |2\pi k|^{-2s})\, \d E \right)\d k\right|^2\\
  \le&\; \int_{\R^d} \left| 1- \int_0^\infty \varphi(E)
    f(E |2\pi k|^{-2s})\, \d E \right|^2\d k \\
  =&\; \frac{d|B_1|}{2s (2\pi)^d} \int_0^\infty
     \left| 1 -\int_0^\infty \varphi(E) f(E t) \,\d E \right|^2
     \frac{\d t}{t^{1+\frac d {2s}}}.
\end{align*}
Thus 
\begin{align*} 
&\sum_{n=1}^N  \int_0^\infty \bigl|u_n^{E+}(x)\bigr|^2 \,\d E  
  \ge \biggl(\int_0^1 \varphi(r)^2\,\d r\biggr)^{-1}\\
  &\rule{2em}{0pt}\times
    \left[  \sqrt{\rho(x)} - \left( \frac{d|B_1|}{2s (2\pi)^d}
        \int_0^\infty 
        \left| 1 -\int_0^\infty \!\varphi(E) f(E t)\, \d E \right|^2
        \frac{\d t}{t^{1+\frac d {2s}}}\right)^{1/2} \right]_+^2.
\end{align*}
Replacing $\varphi(E)\mapsto \ell \varphi (\ell E)$ and optimizing over
$\ell>0$ we get 
\begin{align*} 
  & \sum_{n=1}^N  \int_0^\infty \bigl|u_n^{E+}(x)\bigr|^2 \,\d E
    \ge \rho(x)^{1+\frac {2 s}{d}}
    {\left( \frac{d}{d+2s}\right)}^2
    {\left( \frac{2s}{d+2s}\right)}^{\frac{4s}{d}} \\
  &\times  \biggl(\int_0^1 \varphi(r)^2\,\d r\biggr)^{-1}
    {\left( \frac{d|B_1|}{2s (2\pi)^d}
    \int_0^\infty 
    \left| 1 -\int_0^\infty\! \varphi(E) f(E t) \,\d E \right|^2
    \!\frac{\d t}{t^{1+\frac d {2s}}}\right)}^{-\frac{2s}{d}}.
\end{align*}
Finally, optimizing over $f$ and $\varphi$ we conclude that
\begin{align*} 
  &\sum_{n=1}^N  \int_0^\infty \bigl|u_n^{E+}(x)\bigr|^2 \,\d E\\
  &\ge \rho(x)^{1+\frac {2 s}{d}}   {\left( \frac{d}{d+2s}\right)}^2
    \left( \frac{2s}{d+2s}\right)^{\frac{4s}{d}}
    \left( \frac{|B_1|}{ (2\pi)^d}\right)^{-\frac{2s}{d}}
    \mathcal{C}_{d,s}^{-\frac{2s}{d}}\\
  &= \rho(x)^{1+\frac {2 s}{d}}  K_{d,s}^{\rm cl} \frac{d}{d+2s}
    \left( \frac{2s}{d+2s}\right)^{\frac{4s}{d}}
    \mathcal{C}_{d,s}^{-\frac{2s}{d}} .
\end{align*}
Inserting this bound in \eqref{eq:kinetic-representationfrac} we get the
desired inequality.  
\end{proof}

\subsection{Further results} \label{sec:2.5}
The idea of optimizing momentum decompositions is also useful to improve the
Lieb--Thirring kinetic constant on the sphere and on the torus in
\cite{IlyLapZel-20}, and to improve the constant in the
Cwikel--Lieb--Rozenblum inequality in \cite{HunKunRieVug-18}.  
This technique can be developed to derive new semiclassical inequalities;  see
\cite{FraLewLieSei-11,FraLewLieSei-13} for a positive density analogue of the
Lieb--Thirring inequality.

\section{Lundholm-Solovej method} \label{sec:3}
\subsection{Kinetic inequality via local exclusion principle} \label{sec:3.1}
The Lieb--Thirring inequality \cite{LieThi-75,LieThi-76} was originally
invented to give an energy lower bound for fermionic particles.
From first principles of quantum mechanics, a system of $N$ identical
(spinless) fermions in $\R^d$ can be described by a normalized wave function
$\Psi_N \in L^2((\R^d)^N)$ satisfying    
\begin{equation} \label{eq:Pauli}
  \Psi_N(x_1,\ldots,x_i,\ldots,x_j,\ldots,x_N)
  = - \Psi_N(x_1,\ldots,x_j,\ldots,x_i,\ldots,x_N), \ \forall i\ne j.
\end{equation}
Here ${x_i\in \R^d}$ is the position of the $i$-th particle (we ignore the
spin for simplicity) and $|\Psi_N|^2$ is interpreted as the probability
density of $N$ particles.

The anti-symmetry condition \eqref{eq:Pauli}, also called Pauli's exclusion
principle, implies that two fermionic particles cannot occupy the same
position.
Clearly, $\Psi_N=0$ if $x_i=x_j$ for $i\ne j$. Moreover,  if we define the
one-body density matrix $\gamma_{\Psi_N}$ as an operator on $L^2(\R^d)$ with
kernel 
\begin{align*}
  &\gamma_{\Psi_N}(x,y )\\
  &:= \sum_{j=1}^N \int_{\R^{d(N-1)}}\!\!
    \Psi(\dots,x_{j-1},x,x_{j+1},\dots)
    \overline{\Psi(\dots,x_{j-1},x,x_{j+1},\dots)}  \prod\limits_{i\ne j} \d x_i,
\end{align*}
then we have the operator inequality \cite{LieLos-01,LieSei-10}
\begin{equation} \label{eq:Pauli-gamma}
  0\le \gamma_{\Psi_N} \le 1 \quad \text{ on }L^2(\R^d).
\end{equation}
This is nontrivial since $\Tr \gamma_{\Psi_N}=N$.
Without the anti-symmetry condition \eqref{eq:Pauli}, $\gamma_{\Psi_N}$ may
have an eigenvalue as large as $N$ (in fact,
$\gamma_{\Psi_N}=N |u\rangle \langle u|$ if $\Psi_N=u^{\otimes N}$).
The Lieb--Thirring inequality allows us to bound the kinetic energy of
$\Psi_N$ in terms of its one-body density  
\begin{align*}
  \rho_\Psi(x) &:= \gamma_{\Psi_N}(x,x)\\
  &= \sum_{j=1}^N \int_{\R^{d(N-1)}}
    \bigl|\Psi(x_1,\dots,x_{j-1},x,x_{j+1},\dots,x_N)\bigr|^2
    \prod\limits_{i\ne j} \d x_i.
\end{align*}
\begin{theorem}
\mbox{\normalfont(Lieb--Thirring kinetic inequality for fermions)}
\label{thm:LT-kinetic-fer}
Let $d \ge 1$ and $s>0$.  Let $\Psi_N \in L^2\bigl((\R^d)^N\bigr)$ be a
normalized wave function satisfying the anti-symmetry \eqref{eq:Pauli}.
Then  
\begin{equation} \label{eq:LT-kinetic-fer}
  \biggl\langle \Psi_N,  \sum_{i=1}^N (-\Delta_{x_i})^s \Psi_N\biggr\rangle 
  \ge K_{d,s}  \int_{\R^d} \rho_{\Psi_N}^{1+2/d}(x)\,\d x.  
\end{equation}
with the same constant $K_{d,s}$ in \eqref{eq:LT-ineq-frac}. 
\end{theorem}
Note that  the left side of \eqref{eq:LT-kinetic-fer} is nothing but 
$$\Tr\bigl((-\Delta)^s\gamma_{\Psi_N}\bigr)=\Tr\bigl((-\Delta)^{s/2}
   \gamma_{\Psi_N}(-\Delta)^{s/2}\bigr).$$
Hence, applying \eqref{eq:LT-kinetic-fer} to the Slater determinant
$\Psi_N= u_1 \wedge u_2  \wedge\dots\wedge u_N$ we recover the kinetic
inequality for orthonormal functions in \eqref{eq:LT-ineq-frac}.
On the other hand, we can deduce \eqref{eq:LT-kinetic-fer} from
\eqref{eq:LT-ineq-frac} and the operator bound \eqref{eq:Pauli-gamma} by a
convexity argument (see \cite[Lemma 3]{Frank-20}).
Moreover, the operator bound \eqref{eq:Pauli-gamma} also allow us to deduce
\eqref{eq:LT-kinetic-fer} from the eigenvalue bound in Theorem
\ref{thm:LT-ev-frac}.
Nevertheless, we will discuss an alternative approach to
\eqref{eq:LT-kinetic-fer} below as it will open the way to further developments.  

In 2013, Lundholm and Solovej \cite{LunSol-13,LunSol-14} realized that one can
deduce the Lieb--Thirring inequality \eqref{eq:LT-kinetic-fer} using only a
weaker version of Pauli's exclusion principle \eqref{eq:Pauli}.
More precisely, they need only a rather simple consequence of the operator
inequality \eqref{eq:Pauli-gamma}, which holds for a larger class of quantum
systems than just Fermi gases.
We refer to Lundholm's lecture notes \cite{Lundholm-18} for a pedagogical
introduction to the theory.
In the following we will follow the simplified representation in
\cite{LunNamPor-16}.  

\medskip\noindent%
{\bf Decomposition in position space.}
If $\R^d$ is covered by disjoint domains $\{\Omega\}$, then 
$$
(-\Delta)^s
= \sum_{\Omega} (-\Delta)^s_{|\Omega} \quad \text{ on }L^2(\R^d)
$$
where the Neumann Laplacian $(-\Delta)^s_{|\Omega}$ is  defined via the
quadratic form  
$$
\bigl\langle u, (-\Delta)^s_{|\Omega} u\bigr\rangle_{L^2(\R^d)}
= \|u\|_{\dot{H}^s(\Omega)}^2.
$$
Here the seminorm $\norm{u}_{\dot{H}^s(\Omega)}^2$ is defined as follows, 
$$\norm{u}_{\dot{H}^s(\Omega)}^2
=\left\{
  \begin{array}{l@{\quad\mathrm{if}\ }l}
    \displaystyle\sum_{\abs{\alpha}=s}\frac{s!}{\alpha!}
    \int_\Omega \abs{D^\alpha u(x)}^2\,\d x & s\in \mathbb{N},\\
    \displaystyle
    c_{d,\sigma}\!\sum_{\abs{\alpha}=m}\!\frac{m!}{\alpha!}
    \int_{\Omega}\int_{\Omega}\!
    \frac{\abs{D^\alpha u(x)-D^\alpha u(y)}^2}
         {\abs{x-y}^{d+2\sigma}}\,\d x\,\d y
   & s\not\in\mathbb{N}.
  \end{array}\right.
$$
In the case ${s\not\in\mathbb{N}}$, we have used the notation
${s=m+\sigma}$ with ${m\in \mathbb{N}}$, ${\sigma \in (0,1)}$ and
\begin{displaymath}
  c_{d,\sigma}:=\frac{2^{2\sigma-1}}{\pi^{d/2}}
    \frac{\Gamma(d/2+\sigma)}{|\Gamma(-\sigma)|}.
\end{displaymath}
The coefficient $c_{d,\sigma}$ comes from the well-known formula
(see e.g. \cite[Lemma 3.1]{FraLieSei-07})
\begin{displaymath}
  \bigl\langle u,(-\Delta)^\sigma u\bigr\rangle
    = c_{d,\sigma} \int_{\R^d}\int_{\R^d} \frac{\abs{u(x)-u(y)}^2}
    {\abs{x-y}^{d+2\sigma}}\,\d x\,\d y.  
\end{displaymath}
Consequently, for any $N$-body wave function $\Psi_N \in H^s(\R^{dN})$ we can
decompose  
\begin{equation} \label{eq:localization}
  \cE_{\R^d}[\Psi_N]
  := \biggl\langle \Psi_N,  \sum_{i=1}^N (-\Delta_{x_i})^s
  \Psi_N\biggr\rangle  =  \sum_\Omega {\mathcal E}_\Omega[\Psi_N]
\end{equation}
where the local energy on $\Omega$ is defined by
\begin{equation} \label{eq:def-EQ}
  \cE_\Omega[\Psi_N]
  := \biggl\langle \Psi_N,
  \sum_{j=1}^N (-\Delta_{x_j})^s_{|\Omega} \Psi_N \biggr\rangle.
\end{equation}

So far we have defined the Neumann Laplacian $(-\Delta)^s_{|\Omega}$ on
$L^2(\R^d)$ for convenience, but it can be restricted naturally to an operator
on $L^2(\Omega)$.
It is important that the kernel of this restriction has only finite dimensions. 
\begin{lemma}\mbox{\normalfont(Lower bound on Neumann Laplacian)}
\label{lem:lower-bound-vN-Laplacian} 
Let $d\ge 1$ and $s>0$.
Then for any cube $Q\subset \R^d$ we have the operator inequality on $L^2(Q)$:
\bq \label{eq:Neumann-bound}
  (-\Delta)^s_{|Q} \ge \frac{C_{d,s}}{|Q|^{2s/d}}(1 - P_q)
\eq
for a constant $C_{d,s}>0$ independent of $Q$ and a rank-$q$ projection $P_q$,
where   
$$q := \# \{ \textup{multi-indices}\ \alpha \in \{0,1,\dots\}^d : 0
\le |\alpha|<s \}.$$
\end{lemma}

\begin{proof} If ${s=1}$, then the result is obvious since the eigenvalues of
the Neumann Laplacian on $L^2(Q)$ are given explicitly by
$\bigl\{|Q|^{-2/d}|2\pi k|^2 \ \big| \ k\in  \{0,1,\dots\}^d\bigr\}$;
in particular the  eigenvalue $0$ is single. We refer to \cite[Lemma
11]{LunNamPor-16} for the general case.  
\end{proof}

A consequence of \eqref{eq:Neumann-bound} and \eqref{eq:Pauli-gamma} is  
\begin{lemma}\mbox{\normalfont(Local exclusion for fermions)}
\label{lem:local-exclusion-fermions} 
Let $d\ge 1$ and $s>0$.
Let $\Psi_N$ be a normalized fermionic wave function in $H^s(\R^{dN})$
satisfying 
\eqref{eq:Pauli}.
Then for any cube $Q\subset \R^d$, we have 
\begin{equation} \label{eq:local-exclusion}
  \cE_Q[\Psi_N] \ge  C_{d,s} |Q|^{-2s/d}
  \biggl[ \int_Q \rho_{\Psi_N}^{\ }(x)\, \d x - q \biggr]_+,
\end{equation}
where $\cE_Q[\Psi_N]$ is defined in \eqref{eq:def-EQ} and
$q$ is given in Lemma \ref{lem:lower-bound-vN-Laplacian}.
\end{lemma}
In the non-relativistic case ${s=1}$, this weak formulation of the exclusion
principle was used by Dyson and Lenard in their first proof of the stability
of matter~\cite{DysLen-67}.
As realized in \cite{LunSol-13,LunSol-14}, this can be used as a key tool to
derive Lieb--Thirring inequalities.   

\begin{Proof}[Proof of Lemma \ref{lem:local-exclusion-fermions}]
We use the spectral decomposition 
\bq \label{eq:spec-gamma}
\gamma_{\Psi_N}= \sum_{n\ge 1} \lambda_n |u_n\rangle \langle u_n|
\quad \text { on } L^2(\R^d)
\eq
where $\{u_n\}_{n\ge 1} \subset H^s(\R^d)$ are orthonormal in $L^2(\R^d)$ and
$0\le \lambda_n \le 1$ due to \eqref{eq:Pauli-gamma}.
Then using \eqref{eq:Neumann-bound} we have
\begin{align*}
  \cE_Q[\Psi_N]
  &= \Tr \bigl( (-\Delta)^s_{|Q} \gamma_{\Psi_N} \bigr)
    =  \sum_{n\ge 1} \lambda_n \Bigl\langle u_n, (-\Delta_{x_j})^s_{|Q}
    u_n \Bigr\rangle_{L^2(\R^d)}\\
  &\ge \frac{C_{d,s}}{|Q|^{2s/d}} \sum_{n\ge 1} \lambda_n \bigl\langle
    u_n, \1_{Q}(1- P_q)  \1_{Q} u_n \bigr\rangle_{L^2(\R^d)} \\
  &\ge \frac{C_{d,s}}{|Q|^{2s/d}}  \Biggl[ \sum_{n\ge 1} \lambda_n
    \int_{Q} \bigl|u_n(x)\bigr|^2 \,\d x
    - \Tr (P_q) \Biggr]\\
  &= \frac{C_{d,s}}{|Q|^{2s/d}}
    \biggl[ \int_{Q}\rho_{\Psi_N}^{\ }(x)\,\d x - q \biggr]. 
\end{align*}
Also obviously the left side of \eqref{eq:local-exclusion} is
nonnegative. This completes the proof of \eqref{eq:local-exclusion}.  
\end{Proof}

The second key ingredient to prove the Lieb--Thirring inequality
\eqref{eq:LT-kinetic-fer} is a Gagliardo--Niren\-berg inequality on bounded
domains (see \cite[Lemma 8]{LunNamPor-16}).
This part requires no symmetry condition on the wave functions.  
\begin{lemma}\mbox{\normalfont(Local uncertainty)}
\label{lem:local-uncertainty}
Let $d\ge 1$ and $s>0$. Let $\Psi_N$ be a wave function in $H^s(\R^{dN})$ 
for arbitrary $N \ge 1$ and let $Q$ be an arbitrary cube in $\R^d$.
Then 
\begin{equation} \label{eq:local-uncertainty}
  \cE_Q[\Psi_N] \ge  C_{d,s}
  \frac{\int_Q \rho_{\Psi_N}^{\ }(x)^{1+2s/d}\,\d x}
  {\Bigl(\int_Q \rho_{\Psi_N}^{\ }(x)\,\d x \Bigr)^{2s/d}} 
		- \frac{1}{|Q|^{2s/d}} \int_Q \rho_{\Psi_N}^{\ }(x)\,\d x. 
\end{equation}
\end{lemma}
\begin{Proof}[Sketch of proof]
By translation and dilation, it suffices to prove \eqref{eq:local-uncertainty}
for $Q=[0,1]^d$.  
When $N=1$,  \eqref{eq:local-uncertainty} is equivalent to the
Gagliardo--Niren\-berg inequality
\bq\label{eq:local-uncertainty-Hs-u}
\|u\|_{H^s(Q)}^{\theta} \|u\|_{L^2(Q)}^{1-\theta}
\ge C_{d,s} \|u\|_{L^{q}(Q)},
\quad q=2+ \frac{4s}{d},\quad  \theta=\frac{d}{d+2s}.\qquad
\eq 
By the extension theorem (see \cite[Theorem~7.41]{Adams-75}), it suffices to
prove that  
$$ \|U\|_{\dot H^s(\R^d)}^{\theta} \|U\|_{L^2(\R^d)}^{1-\theta} 
\ge C_{d,s} \|U\|_{L^{q}(\R^d)},
\quad q=2+\frac{4s}{d},\quad\theta=\frac{d}{d+2s}
$$
which follows from Sobolev's embedding theorem. 
	
For $N\ge 1$, we can use the spectral decomposition \eqref{eq:spec-gamma} and
write 
\begin{displaymath}
  \cE_Q[\Psi_N] = \Tr \bigl( (-\Delta)^s_{|Q} \gamma_{\Psi_N}  \bigr)
  =  \sum_{n\ge 1} \| v_n\|_{\dot{H}^s (Q)}^2
\end{displaymath}
with $v_n= \lambda^{1/2} u_n$.
Using H\"older's inequality (for sums), the one-body inequality
\eqref{eq:local-uncertainty-Hs-u} and the triangle inequality we get 
\begin{eqnarray*}
  &&\left( \int_Q \rho_{\Psi}^{\ }(x)\,\d x\right)^{\frac{2s}{d+2s}}
     \left( \cE_Q[\Psi_N]
    +\int_Q \rho_{\Psi}^{\ }(x)\,\d x  \right)^{\frac{d}{d+2s}}\\
  &=& \left(\sum_{n\ge 1}\int_Q\bigl|v_n(x)\bigr|^2\,\d x\right)^{\frac{2s}{d+2s}}
       \left(\sum_{n\ge 1} \|v_n\|_{{H}^s(Q)}^2 \right)^{\frac d {d+2s}} \\
  &\ge& \sum_{n\ge 1}  \|v_n\|_{L^2(Q)}^{\frac{4s}{d+2s}}\,
       \|v_n\|_{{H}^s(Q)}^{\frac{2d}{d+2s}}  \ge C_{d,s}
       \sum_{n\ge 1}  \bigl\| |v_n|^2 \bigr\|_{L^{1+2s/d}(Q)} \\
  &\ge&  C_{d,s} \Bigl\| \sum_{n\ge 1} |v_n|^2 \Bigr\|_{L^{1+2s/d}(Q)}
       = C_{d,s}\bigl\| \rho_\Psi^{\ } \bigr\|_{L^{1+2s/d}(Q)}.
\end{eqnarray*}
This is equivalent to  \eqref{eq:local-uncertainty}.
\end{Proof}

The third key ingredient to prove the Lieb--Thirring inequality
\eqref{eq:LT-kinetic-fer} is a covering lemma, which allows to combine the
local exclusion and local uncertainty principles in an efficient way.  
The following is a simplified version of \cite[Lemma 12]{LunNamPor-16}.
\begin{lemma}\mbox{\normalfont(Covering lemma)}
\label{lem:covering}
Let $0\le f\in L^1(\R^d)$ be a function with compact support.
Take $0<\Lambda< \int_{\R^d} f(x)\,\d x$.
Then we can cover $\R^d$ by a collection of disjoint cubes $\{Q\}$ such that 
\begin{equation} \label{eq:covering-0}
  \int_{Q} f(x)\,\d x \le \Lambda, \quad \forall Q
\end{equation}
and there exists $C_{\alpha,q}>0$ for every $\alpha>0$ and $0<q \le (1-\eps)
\Lambda 2^{-d}$ such that 
\begin{equation} \label{eq:covering}
  \sum_{Q} \frac{1}{|Q|^{\alpha}}
  \Biggl( \biggl[\int_{Q} f(x)\,\d x - q \biggr]_+ 
- \eps (1- 2^{-\alpha d} ) 4^{-d} \int_{Q} f(x)\,\d x \Biggr) \ge 0.
\end{equation}
\end{lemma}

\begin{proof} Frist, we cover $\supp f$ by a big cube $Q_0$.
Then we divide $Q_0$ into $2^{d}$ disjoint sub-cubes of half-length side. For
every sub-cube $Q$, 
\begin{itemize}
\item If $\int_Q f<\Lambda$, then we stop dividing $Q$.
\item If $\int_Q f\ge \Lambda$, then we continue dividing $Q$ into $2^{d}$
  disjoint sub-cubes and iterate.  
\end{itemize}
This procedure stops after finitely many steps (since $f$ is
integrable) and we obtain a division of $Q_0$ into finitely  any cubes
$Q$'s. We can distribute all these cubes into disjoint groups $\{\cF\}$ such
that in each group~$\cF$: 
\begin{itemize}
\item There exists a smallest cube in $\cF$ such that $\int_Q f \ge
  2^{-d}\Lambda$. 
\item There are at most  $2^d$ cubes of every volume level. 
\end{itemize}
Now we consider each group $\cF$. By the first property of $\cF$, we can find
a smallest cube $Q_m\in \cF$ with $|Q_m|=m$ and  
\begin{align*}
  \sum_{Q \in \cF} \frac{1}{|Q|^{\alpha}}
  \biggl[\int_{Q} f(x)\,\d x - q \biggr]_+
&\ge  \frac{1}{|Q_m|^{\alpha}} \biggl[\int_{Q_m}\! f(x)\,\d x - q \biggr]_+\\
&\ge \frac{1}{m^{\alpha}} (2^{-d}\Lambda-q)
\ge \frac{\eps}{2^d} \frac{\Lambda}{m^{\alpha}}.
\end{align*}
On the other hand, by the second property of $\cF$, there are at most $2^d$
cubes in $\cF$ of volume $2^{kd}m$ for each $k=0,1,\dots$
Moreover, $\int _Q f<\Lambda$ for every cube $Q$.
Hence, 
\begin{align*}
  \sum_{Q\in \cF} \frac{1}{|Q|^\alpha}  \int_Q f(x)\,\d x
  \le \sum_{k\ge 0} \frac{2^d}{(2^{kd}m)^\alpha} \Lambda
  =  \frac{2^d}{1- 2^{-\alpha d}}  \frac{\Lambda}{m^\alpha}. 
\end{align*}
Thus
$$
\sum_{Q \in \cF} \frac{1}{|Q|^{\alpha}}
\Biggl( \biggl[\int_{Q} f(x)\,\d x - q \biggr]_+ 
- \eps (1- 2^{-\alpha d} ) 4^{-d}
\int_{Q} f(x)\,\d x \Biggr) \ge 0.
$$
Summing over all groups $\cF$'s we get the desired inequality. 
\end{proof}

Now we are ready to give an alternative proof of the Lieb--Thirring inequality
\eqref{eq:LT-kinetic-fer}. 
\begin{proof} Let $q$ be as in Lemma \ref{lem:local-exclusion-fermions} and
let $\Lambda=2^{d+1} q$.  
If $N\le \Lambda$, then the desired bound follows immediately from
\eqref{lem:local-uncertainty} by taking $Q\to \R^d$.
If $N> \Lambda$, then by a standard density argument we can assume that
$\Psi_N\in C_c^\infty(\R^d)$.
Then we apply Lemma \ref{lem:covering} with $f=\rho_{\Psi_N}$, $\alpha=2s/d$
and obtain a collection of disjoint cubes $\{Q\}$ covering $\supp
\rho_{\Psi_N}$.  

From the local uncertainty \eqref{eq:local-uncertainty} and
\eqref{eq:covering-0} we have 
\begin{displaymath}
  \cE_{\R^d}[\Psi_N] =\sum_{Q} \cE_{Q}[\Psi_N] \\
  \ge \frac{C_{d,s}}
{ \Lambda^{2s/d}} \int_{\R^d} \rho_{\Psi_N}^{1+2s/d}
- \sum_{Q} \frac{1}{|Q|^{2s/d}} \int_Q \rho_{\Psi_N}^{\ } 
\end{displaymath}
On the other hand, by the local exclusion \eqref{eq:local-exclusion}, for all
$L>0$ we have 
$$
L \cE_{\R^d}[\Psi_N] 
=L \sum_{Q} \cE_{Q}[\Psi_N] \ge \sum_{Q}\frac { L C_{d,s} } { |Q|^{2s/d} }
\biggl[ \int_Q \rho_{\Psi_N}^{\ } - q \biggr]_+.
$$
If we choose $L=L_{d,s}$ such that
$$
\frac{1}{L_{d,s}C_{d,s}}= \frac{1}{2} (1- 2^{-\alpha d} ) 4^{-d},
$$
then \eqref{eq:covering-0} gives us 
$$
 \sum_{Q} \frac 1 { |Q|^{2s/d} }  \Big( L_{d,s} C_{d,s} 
\biggl[ \int_Q \rho_{\Psi_N}^{\ } - q \biggr]_+  
-  \int_Q \rho_{\Psi_N}^{\ } \Big) \ge 0. 
$$
Thus we conclude that 
$$
(1+L_{d,s}) \cE_{\R^d}[\Psi_N]  
\ge \frac{C_{d,s}}{ \Lambda^{2s/d}} \int_{\R^d} \rho_{\Psi_N}^{1+2s/d}.
$$
This completes the proof of \eqref{eq:LT-kinetic-fer}.
\end{proof}

\subsection{Kinetic inequality with semiclassical constant and error term}
\label{sec:3.2}
A natural impression from the proof in Section \ref{sec:3.2} is that it 
involves several non-optimal estimates and potentially gives a rather bad
control on the constant. 
However, it turns out that this proof can be modified to achieve the
semiclassical constant, up to an error which is normally small in
applications. 
Here we focus only on the non-relativistic case $s=1$. 
The following result is taken from \cite{Nam-18}. 

\begin{theorem}  
\label{thm:LT-kin-semi}
Let ${d \ge 1}$.
For any ${N\ge 1}$, let ${\{u_n\}_{n=1}^N \subset H^1(\R^d)}$ be orthonormal
functions in $L^2(\R^d)$ and define the density 
$\rho(x)=\sum_{n=1}^N\bigl|u_n(x)\bigr|^2$. 
Then for all $\eps>0$
\begin{eqnarray}
\label{eq:LT-kin-semi}
\sum_{n=1}^N \int_{\R^d} \bigl|\nabla u_n(x)\bigr|^2\,\d x 
&\ge& (1-\eps) K_d^{\rm cl} 
\int_{\R^d}  \rho(x)^{1+\frac {2}{d}}\,\d x\nonumber\\
&&\mbox{}- \frac{C_d}{\eps^{3+4/d} }  
\int_{\R^d} \left|\nabla \sqrt{\rho(x)}\right|^2\,\d x .
\end{eqnarray}
\end{theorem}
Note that our bound \eqref{eq:LT-kin-semi} implies the  Lieb--Thirring
inequality \eqref{eq:intro-LT-ineq-kin} with a non-sharp constant, thanks to
the Hoffmann--Ostenhof inequality \cite{HO-77} (or the diamagnetic inequality) 
\bq \label{eq:HO}
\sum_{n=1}^N \int_{\R^d} \bigl|\nabla u_n(x)\bigr|^2  
\ge \int_{\R^d} \left|\nabla \sqrt{\rho(x)}\right|^2 \d x.
\eq
Moreover, in many applications, the gradient term is much smaller than the
kinetic energy. 
For example, at the ground state of $N$ ideal (i.e. non-interacting) fermions
in a fixed volume, the gradient term  is proportional to $N$ while the kinetic
energy grows as $N^{1+2/d}$. 
The gradient  terms have also appeared in recent improvements
\cite{BenBleLos-12,LewLie-15} of the Lieb--Oxford estimate on Coulomb exchange
energy \cite{LieOxf-80}. 

Recall that when $d\le 2$, we know that the optimal value of $K_{d}$ is
strictly smaller than $K_d^{\rm cl}$. 
On the other hand,  our bound \eqref{eq:LT-kin-semi} holds for all $d\ge 1$,
so the additional error term is unavoidable.  

\begin{Proof}[Sketch of proof]  \\
{\bf Step 1.} For every cube $Q\subset \R^d$ and every $\mu>0$ we can write
\begin{align*}
\sum_{n\ge 1} \int_Q |\nabla u_n|^2 
&=  \Tr_{L^2(Q)} \bigl( (-\Delta_Q -\mu ) \gamma_Q\bigr) + \mu \int_{Q} \rho\\ 
&\ge \Tr_{L^2(Q)}  (-\Delta_Q -\mu )_- + \mu \int_{Q} \rho  
\end{align*}
with $-\Delta_Q$ the Neumann Laplacian on $L^2(Q)$ 
and $\gamma_Q=\sum\limits_{n\ge 1} |\1_Q u_n\rangle \langle \1_Q u_n| $ 
satisfying $0\le \gamma_Q\le 1$ on $L^2(\R^d)$. Using the explicit eigenvalues
of the Neumann Laplacian, we have 
\begin{align*}
\Tr_{L^2(Q)} (-\Delta_Q -\mu )_-
  &= |Q|^{-2/d}  \hspace*{-0.8em}\sum_{p\in \{0,1,\dots\}^d} 
\hspace*{-1em}\bigl[\pi^2 |p|^2-\mu\bigr]_- \\
  &\ge  |Q|^{-2/d} \Bigl(-L_{1,d}^{\rm cl} \mu^{1+2/d}
                        - C_d (\mu^{1+1/d}+1)\Bigr). 
\end{align*}
Optimizing over $\mu$ we obtain 
\begin{equation} \label{eq:local-LT}
\sum_n \int_Q \bigl|\nabla u_n(x)\bigr|^2\, \d x 
\ge K_{d}^{\rm cl} |Q|^{-2/d}\left[\left(\int_Q \rho  \right)^{1+2/d} 
-C  \left(\int_Q \rho  \right)^{1+1/d}  \right].
\end{equation}
{\bf Step 2.} By a density argument, we can assume $\{u_n\}\subset
C_c^\infty(\R^d)$ and hence $\rho$ has compact support. 
As in the covering lemma, we can cover $\supp \rho$ by disjoint cubes $\{Q\}$
such that  
\bq \label{eq:cv-1}
\int_Q \rho_\gamma\le \Lambda, \quad \forall Q,
\eq 
and
\bq \label{eq:cv-2}
 \sum_Q  |Q|^{-2/d}\left[C\Lambda^{-1/d}
\left(\int_Q \rho_\gamma \right)^{1+2/d} 
-  \left(\int_Q \rho_\gamma \right)^{1+1/d}  \right] \ge 0.
\eq 
{\bf Step 3.} By Poincar\'e's inequality we can show that for every $\eps>0$, 
\begin{eqnarray*}
\frac 1 {|Q|^{2/d}}  \left( \int_{Q} \rho_\gamma \right)^{1+2/d} 
&\ge& \frac 1 {(1+\eps)^{(1+4/d)}} 
\int_Q \rho_\gamma^{1+2/d}\\&&\mbox{} - \frac{C}{\eps^{(1+4/d)} }
\left( \int_Q |\nabla \sqrt{\rho_\gamma}|^2 \right) 
\left(\int_Q \rho_\gamma \right)^{2/d}.  
\end{eqnarray*}
Combining the latter bound with \eqref{eq:cv-1}, \eqref{eq:cv-2} and
\eqref{eq:local-LT}, we obtain 
\begin{eqnarray*}
\sum_Q \int_Q |\nabla u_n|^2 
&\ge& K_{d}^{\rm cl} (1-C\Lambda^{-1/d}) 
\sum_{Q} |Q|^{-2/d} \left(\int_Q \rho_\gamma \right)^{1+2/d}\\
&\ge& K_{d}^{\rm cl} (1-C\Lambda^{-1/d})  
\biggl[ \frac 1 {(1+\eps)^{(1+4/d)}}\int_{\R^d} \rho_\gamma^{1+2/d}\\ 
&&\rule{8em}{0pt}\mbox{} -\frac{ C \Lambda^{2/d} }{\eps^{(1+4/d)}} 
\int_{\R^d} |\nabla \sqrt{\rho_\gamma}|^2 \biggr].
\end{eqnarray*}
Taking $\eps=\Lambda^{-1/d}$ and optimizing over $\Lambda$ we get the desired
result.  
\end{Proof}

Theorem \ref{thm:LT-kin-semi} can be seen as a first step towards the local
density approximation for many-body quantum systems. 
More precisely, for any $N\ge 1$ and $\rho\ge0$ with $\int_{\R^d}\rho=N$, we
can define the Levy--Lieb energy functional \cite{Levy-79,Lieb-83} for the
kinetic operator 
$$
\cK(\rho) = \inf_{\rho_{\Psi_N}^{\ }=\rho } \biggl\langle \Psi_N , 
\sum_{i=1}^N(-\Delta_{x_i}) \Psi_N \biggr\rangle_{L^2(\R^{dN})}.
$$
Here the infimum is taken over all (normalized) fermionic wave functions in
$L^2(\R^{dN})$ whose one-body density is exactly equal to $\rho$. 
Then \eqref{eq:LT-kin-semi} is equivalent to the lower bound
\begin{equation} \label{eq:Kp-lower}
\cK(\rho) \ge (1-\eps) K_d^{\rm cl} \int_{\R^d}  
\rho(x)^{1+\frac {2}{d}}\,\d x - \frac{C_d}{\eps^{3+4/d} }  
\int_{\R^d} \left|\nabla \sqrt{\rho(x)}\right|^2\,\d x, \quad \forall \eps>0. 
\end{equation}
Of course, by the Lieb--Thirring conjecture one expects that both the gradient
term and the $\eps$ dependence can be removed when $d\ge 3$. 

On the other hand, it is conjectured \cite{MarYou-58,Lieb-83} that $\cK(\rho)$
satisfies the upper bound  
\bq \label{eq:Kp-upper-conj}
\cK(\rho) \le  K_d^{\rm cl} \int_{\R^d}  \rho(x)^{1+\frac {2}{d}} \,\d x
+  \int_{\R^d} \left|\nabla \sqrt{\rho(x)}\right|^2\,\d x.
\eq
The appearance of the gradient term on the right side of
\eqref{eq:Kp-upper-conj} is reasonable since the kinetic energy cannot be
controlled by an integral of $\rho$ alone due to the Hoffmann--Ostenhof
inequality \eqref{eq:HO}. 
Recently, Lewin, Lieb and Seiringer proved in \cite{LewLieSei-20} that, for a
grand-canonical analogue $\widetilde{\cK} (\rho)$ of $\cK(\rho)$, 
\begin{equation} \label{eq:Kp-upper}
\widetilde{\cK} (\rho) \le (1+\eps)K_{d}^{\rm cl}
 \int_{\R^d} \rho(x)^{1+\frac 2 d} \,\d x
+ \frac{C}{\eps}\int_{\R^d} \left|\nabla \sqrt{\rho(x)}\right|^2\,\d x, 
\quad \forall \eps>0.
\end{equation}
A result weaker than \eqref{eq:Kp-upper} was used in \cite{GotNam-18} in
the context of proving Gamma-convergence of the Levy--Lieb model
to Thomas--Fermi theory. 
Removing $\eps$ in both \eqref{eq:Kp-lower} and \eqref{eq:Kp-upper} is
interesting and difficult. 
Nevertheless, in the current form, they are already useful to justify the
local density approximation in certain regimes; see
\cite{LewLieSei-20,LewLieSei-19} for further details.

\subsection{Kinetic inequality for functions vanishing on diagonal set} \label{sec:3.3}

Recall that Pauli's exclusion principle \eqref{eq:Pauli} implies that the wave
function $\Psi_N$ vanishes on the diagonal set of $(\R^d)^N$, namely 
\bq \label{eq:dia-vanish}
\Psi(x_1,\dots,x_N)=0\quad\text{ if $x_i=x_j$ for some $i\ne j$}. 
\eq
Thus a natural question is whether the Lieb--Thirring inequality
\eqref{eq:LT-kinetic-fer} remains valid if \eqref{eq:Pauli} is replaced by the
weaker condition \eqref{eq:dia-vanish}. 
This question is nontrivial since \eqref{eq:dia-vanish} is not sufficient to
ensure the operator inequality \eqref{eq:Pauli-gamma}. 
The following answer is taken from \cite{LarLunNam-19}. 
\begin{theorem}{\normalfont(Lieb--Thirring inequality for wave functions
    vanishing on diagonals)} \label{thm:kin-vanish}
Let $d\ge 1$ and $s>0$. 
Let $N\ge 1$ and let $\Psi_N\in C_c^\infty(\R^{dN})$ be a normalized wave
function satisfying \eqref{eq:dia-vanish}. 
Then we have the Lieb--Thirring inequality
\begin{equation} \label{eq:LT-diagonal}
  \biggl\langle \Psi_N,  \sum_{i=1}^N (-\Delta_{x_i})^s\Psi_N \biggr\rangle
  \ge C_{d,s}\int_{\R^d} \rho_{\Psi_N}^{\ }(x)^{1+2s/d}\,\d x 
\end{equation}
with a constant $C_{d,s}>0$ independent of $N$ and $\Psi_N$ if and only if
$2s>d$.  
\end{theorem}
Actually, the condition ${2s>d}$ is related to the Sobolev embedding
$H^s(\R^d)\subset C(\R^d)$. 
Heuristically, this is the minimum condition for \eqref{eq:dia-vanish} to be
``nontrivial,'' otherwise \eqref{eq:LT-diagonal} must fail. 
We refer to  \cite{LarLunNam-19} for a detailed discussion on this negative
direction, and below let us focus only on the derivation of
\eqref{eq:LT-diagonal} when $2s>d$.  

From the general strategy in Section \ref{sec:3.1}, it suffices to derive a
local exclusion bound similar to \eqref{eq:local-exclusion}.  

\begin{Proof}[Sketch of the local exclusion bound]\\  
{\bf Step 1.} First, note that \eqref{eq:local-exclusion} can be reduced to a
simpler estimate:  for scale-covariant systems, such a local exclusion simply
boils down to the strict positivity of the local energy. 
This idea is inspired by Lundholm and Seiringer \cite{LunSei-18} and seems
very helpful for future applications. 
The following abstract formulation is taken from \cite[Lemma 4.1]{LarLunNam-19}. 
\begin{lemma}\mbox{\normalfont(Covariant energy bound)}
\label{lem:cov-energy} 
Assume that to any $n\in \mathbb{N}_0$ and any cube $Q\subset \R^d$
there is associated a non-negative number (`energy') $E_n(Q)$
satisfying the following properties, for some constant $s>0$: 
\begin{itemize}
\item[{\rm (i)}] (scale-covariance)
	$E_n(\lambda Q) = \lambda^{-2s} E_n(Q)$ for all $\lambda > 0$;
\item[{\rm (ii)}] (translation-invariance)
	$E_n(Q+x) = E_n(Q)$ for all $x \in \R^d$;
\item[{\rm (iii)}] (superadditivity)
     For any collection of disjoint cubes $\{Q_j\}_{j=1}^J$ such that their
     union is a cube,  
 $$E_n\Bigl(\bigcup_{j=1}^J Q_j\Bigr) 
  \ge \min_{\{n_j\} \in \mathbb{N}_0^J \, s.t.\, \sum_j n_j = n}\ 
  \sum_{j=1}^J E_{n_j}(Q_j) ;$$
\item[{\rm (iv)}] (a priori positivity) 
There exists $q \ge 0$ such that $E_n(Q) > 0$ for all $n \ge q$.
\end{itemize}
Then there exists a constant $C>0$ independent of $n$ and $Q$ such that
\bq \label{eq:cov-energy}
   E_n(Q) \ge C |Q|^{-2s/d} n^{1+2s/d}, \quad \forall n\ge q.
\eq
\end{lemma}
{\bf Step 2.} The above abstract result applies to the  local energy 
\bq \label{eq:def-EnQ}
  E_N(Q) := \inf  \biggl\langle \1_{\Omega^N}\Psi_N, 
  \sum_{j=1}^N (-\Delta_{x_j})^s_{|\Omega} \1_{\Omega^N} \Psi_N 
  \biggr\rangle_{L^2(\Omega^N)}  
\eq
where the infimum is taken over all wave functions $\Psi_N \in
C_c^\infty(\R^d)$ satisfying \eqref{eq:dia-vanish} and normalized
$\|\Psi_N\|_{L^2(Q^N)}=1$. 
Note that the right side of \eqref{eq:def-EnQ} is in different from
$\cE_Q[\Psi_N]$ in \eqref{eq:def-EQ} because only the ``completely localized
energy'' in \eqref{eq:def-EnQ} satisfies the superadditivity in Lemma
\ref{lem:cov-energy} (iii). 
Moreover, the conditions in (i) and (ii) obviously hold. 

The key assumption \eqref{eq:dia-vanish} is used to derive the strict
positivity in Lemma \ref{lem:cov-energy} (iv). 
This is nontrivial.
The central facts used in the proof is that the kernel of the associated
Neumann Laplacian must be a polynomial (of many variables), and that if a
polynomial vanishes on too many diagonals then it must be zero. We refer to
\cite[Theorem 5.1]{LarLunNam-19} for details. 

\noindent {\bf Step 3.} Finally, using  \eqref{eq:cov-energy} and a many-body
localization technique, we can deduce the desired local exclusion bound 
$$\cE_Q[\Psi_N] \ge C |Q|^{-2s/d} \Big[ \int_Q \rho_{\Psi_N} - q\Big]_+$$
with the same constants $C,q$ in \eqref{eq:def-EnQ} and with $\cE_Q[\Psi_N]$
defined in \eqref{eq:def-EQ}.
See \cite[Lemma 4.4]{LarLunNam-19} for details. 
\end{Proof}

\subsection{Lieb--Thirring inequality for interacting systems} \label{sec:3.4}
In order to obtain an exclusion bound, instead of putting a condition on wave
functions, one can also add a repulsive term to the Hamiltonian.
Given the kinetic operator $(-\Delta)^s$, it is natural to consider the interaction
potential $w(x)=|x|^{-2s}$ which has the same scaling property.
This leads to the following 
\begin{theorem}\mbox{\normalfont(Lieb--Thirring inequality for interacting systems)}
  \label{thm: LT-int-1}\\
Let ${d\ge 1}$, ${s>0}$ and ${\lambda>0}$.
For any $\Psi_N\in H^s(\R^{dN})$ which is normalized in $L^2(\R^{dN})$, we have
\begin{eqnarray} \label{eq-LT-potential-s}
  &&\left\langle \Psi_N, \left(\sum_{i=1}^N (-\Delta_{x_i})^s
     + \sum_{1 \leq i<j \leq N} \frac{\lambda}{\abs{x_i-x_j}^{2s}} \right)
     \Psi_N \right\rangle\nonumber\\ 
  &\geq& C_{\rm LT}(d,s,\lambda)
         \int_{\R^d} \rho_{\Psi_N}^{\ }(x)^{1+\frac{2s}{d}}\,\d x.
\end{eqnarray}
The constant $C_{\rm LT}(d,s,\lambda)>0$ is independent of $N$ and $\Psi_N$. 
\end{theorem}	
This result was first proved by Lundholm and Solovej for ${s=1}$ and ${d=1}$
in \cite{LunSol-13a}.
The extension to  ${s=1}$, ${d>1}$ was done by Lundholm, Portmann, and
Solovej in  \cite{LunPorSol-15}.
The general power ${s>0}$ was treated in \cite{LunNamPor-16}.
The proof in \cite{LunNamPor-16} is based on the strategy in Section
\ref{sec:3.1}, but now the local exclusion bound is derived from the
interaction.  
\begin{lemma}\mbox{\normalfont(Local exclusion by interaction)}
  \label{lem:local-ex-int} 
For all $d\geq 1$, $s > 0$, for every normalized function $\Psi\in
L^2(\R^{dN})$ and for an arbitrary collection of disjoint cubes $Q$'s in
$\R^d$, one has 
$$
\left\langle \Psi , \sum_{1\le i<j \le N}
  \frac{1}{|x_i-x_j|^{2s}} \Psi \right \rangle
\ge \sum_Q \frac{1}{2d^s|Q|^{2s/d}} \left[ \biggl( \int_Q \rho_\Psi \biggr)^2
  - \int_Q \rho_\Psi \right]_+.
$$
\end{lemma}
\begin{proof} This result follows from the operator estimate
\begin{align*}  
  \sum_{1\le i<j \le N} \frac{1}{|x_i-x_j|^{2s}}
  &\ge \sum_Q \sum_{1\le i<j \le N}
    \frac{\1_Q(x_i)\1_Q(x_j) }{d^s|Q|^{2s/d}} \\
  &= \sum_Q \frac{1}{2d^s|Q|^{2s/d}}
    \left[\left(\sum_{i=1}^N \1_Q(x_i) \right)^2
    - \sum_{i=1}^N \1_Q(x_i) \right] 
\end{align*}
and the Cauchy--Schwarz inequality.
\end{proof}
	
As explained in \cite{LunNamPor-16}, if $2s<d$, then the Lieb--Thirring
inequality \eqref{eq-LT-potential-s} can be also  derived from the one-body
interpolation inequality  
\begin{align} \label{eq:LT-inter-one-body}
&\;\bigl\langle u, (-\Delta)^s u \bigr\rangle_{L^2(\R^d)}^{1-\frac{2s}{d} } 
  \biggl( \iint_{\R^d \times \R^d}\! \frac{|u(x)|^2\,|u(y)|^2}{|x-y|^{2s}}
  \,\d x\,\d y \biggr)^{\frac{2s}{d}}\nonumber\\
  \ge&\;  C \int_{\R^d} |u(x)|^{2(1+\frac{2s}{d})}\,\d x.
\end{align} 
The inequality \eqref{eq:LT-inter-one-body} was first proved in
\cite{BelOzaVis-11} for ${s=1/2,d=3}$, and extended to all ${0<s<d/2}$ in
\cite{BelFraVis-14} (see \cite{BelGhiMerMorSch-18} for further results in
this direction).
The existence of optimizers for \eqref{eq:LT-inter-one-body} is an interesting
open problem; see \cite{BelFraVis-14} for related discussions.  

It was conjectured in \cite{LunNamPor-16} that in the strong coupling limit
$\lambda \to \infty$, the optimal constant $C_{\rm LT}(s,d,\lambda)$ converges
to the Gagliardo--Nirenberg constant 
\begin{align*}
	C_{\rm GN} (d,s):= \inf_{\substack{u \in H^s(\R^d)\\ 
  \norm{u}_{L^2} = 1}}
  \frac{\bigl\langle u, (-\Delta)^s u \bigr\rangle}
       {\int_{\R^d} |u|^{2 (1+\frac{2s}{d})}}.
\end{align*}
This was proved recently in \cite{KogNam-20}. 
  
\begin{theorem}
{\normalfont(Lieb--Thirring constant in the strong--coupling limit)}
  \label{thm:LT-strong}\\
For any $d\ge 1$ and $s>0$, we have 
 $$
\lim_{\lambda\to \infty} C_{\rm LT}(s,d,\lambda) = C_{\rm GN}(d,s). 
$$	
\end{theorem}	
The heuristic idea behind Theorem \ref{thm:LT-strong} is that in the
strong--coupling limit, each particle is forced to stay away from the others
and the many-body interacting problem reduces to a one-body non-interacting
system.
However, proving this is nontrivial since we have to prove
estimates uniformly in the number of particles.  

The proof of Theorem \ref{thm:LT-strong} in \cite{KogNam-20} is based on a new
construction of covering sub-cubes. 
In Lemma \ref{lem:covering}, the division into sub-cubes follows by a standard
``stopping time argument:''
any cube $Q$ with the mass $\int_Q \rho_\Psi$ bigger than a given quantity
will be divided into $2^d$ sub-cubes.
Consequently, the masses in final sub-cubes may differ up to a factor $2^d$,
leading to a similar factor loss in the Lieb--Thirring constant.
In \cite{KogNam-20}, the stopping time argument is applied to ``clusters of
cubes'' rather than to individual cubes.
At the end, each cluster has essentially at most one particle, allowing us to
recover the constant $C_{\rm GN}(d,s)$ by using  a refined version of the
local uncertainty principle \eqref{eq:local-uncertainty}.
The localization error is compensated by the interaction energy.
This localization argument seems very flexible and may be useful in other
contexts.

\subsection{Further results} \label{sec:3.5} The method represented in Section
\ref{sec:3.1} was originally invented to derive Lieb--Thirring inequalities
for anyons (particles satisfying only some fractional statistics between
bosons and fermions).
See \cite{LunSol-13a,LunSol-13,LunSol-14,LarLun-18,LunSei-18,LunQva-20}
for various results in this direction.   
A similar method was used to prove a Lieb--Thirring inequality for fermionic
particles with point interactions in \cite{FraSei-12}.  

This method is also useful to recover the Hardy--Lieb--Thirring inequality in
\cite{EkhFra-06,FraLieSei-07,Frank-09}, where the kinetic operator is replace
by $(-\Delta)^s-\mathcal{C}_{d,s} |x|^{-2s}$ with $\mathcal{C}_{d,s}$ is the
optimal constant in Hardy's inequality \cite{Herbst-77}.
This requires $2s<d$.
The results in Theorems \ref{thm: LT-int-1}, \ref{thm:LT-strong} were also
extended to the case of Hardy operator; see \cite{LunNamPor-16,KogNam-20}.
The key additional ingredient is the refined Hardy inequality: for all
$s>t>0$ and $\ell>0$, 
$$
(-\Delta)^s - \frac{\mathcal{C}_{s,d}}{|x|^{2s}}
\ge  \ell^{s-t}(-\Delta)^t -C_{d,s,t}\ell^{s} \quad \text{on}~ L^2(\R^d).
$$
This bound was first proved for $s=1/2$, $d=3$ by Solovej, S\o rensen, and
Spitzer  \cite[Lemma 11]{SolSorSpi-10} and then generalized to the full range
$0<s<d/2$ by Frank \cite[Theorem 1.2]{Frank-09}.  

Recently, the fermionic Hardy--Lieb--Thirring inequality has been extended to
include fractional Pauli operators in \cite{BleFou-19}.
It is unclear whether the approach in this section could be adapted to study
this case.  

\medskip

\noindent\textbf{Acknowledgments.}
I would like to thank Rupert Frank, Simon Larson, and Douglas Lundholm for
helpful remarks on a preliminary version of this note.

\end{document}